\documentclass[journal]{IEEEtran}
\IEEEoverridecommandlockouts

\usepackage[ruled,linesnumbered]{algorithm2e}
\usepackage{cite}
\usepackage{amsmath,amssymb,amsfonts}
\usepackage{algorithmic}
\usepackage{graphicx}
\usepackage{textcomp}
\usepackage{xcolor}
\usepackage{enumerate}
\usepackage{booktabs}
\usepackage{diagbox}
\usepackage{supertabular}
\usepackage{caption}
\usepackage{float} 
\usepackage{verbatim}
\usepackage{bm} 
\usepackage{breqn}
\usepackage{stfloats}
\usepackage{amsthm}
\usepackage{hyperref}
\usepackage{makecell}
\usepackage{amssymb}
\usepackage{subfigure}
\usepackage{multirow}
\usepackage{array}

\newtheorem{theorem}{Theorem}

\makeatletter
\newcommand{\rmnum}[1]{\romannumeral #1}
\newcommand{\Rmnum}[1]{\expandafter\@slowromancap\romannumeral #1@}
\makeatother

\def\BibTeX{{\rm B\kern-.05em{\sc i\kern-.025em b}\kern-.08em
    T\kern-.1667em\lower.7ex\hbox{E}\kern-.125emX}}

\begin{document}
\title{Blockchain-based Pseudonym Management for Vehicle Twin Migrations in Vehicular Edge Metaverse}

\author{Jiawen Kang, Xiaofeng Luo, Jiangtian Nie, Tianhao Wu, Haibo Zhou, Yonghua Wang*, \\Dusit Niyato, \textit{Fellow, IEEE}, Shiwen Mao, \textit{Fellow, IEEE}, and Shengli Xie, \textit{Fellow, IEEE}
\thanks{
J. Kang, X. Luo, T. Wu, Y. Wang, and S. Xie are with the School of Automation, Guangdong University of Technology, China (e-mail: kavinkang@gdut.edu.cn; gdutxiaofengluo@163.com; wutianhao32@163.com; wangyonghua@gdut.edu.cn; shlxie@gdut.edu.cn).

J. Nie and D. Niyato are with the School of Computer Science and Engineering, Nanyang Technological University, Singapore (e-mail: dniyato@ntu.edu.sg; jnie001@e.ntu.edu.sg).

H. Zhou is with the School of Electronic Science and Engineering, Nanjing University, China (e-mail: haibozhou@nju.edu.cn).

S. Mao is with the Department of Electrical and Computer Engineering, Auburn University, USA (e-mail: smao@ieee.org).

(\textit{*Corresponding author: Yonghua Wang})
}
}

\maketitle

\begin{abstract}
Driven by the great advances in metaverse and edge computing technologies, vehicular edge metaverses are expected to disrupt the current paradigm of intelligent transportation systems. As highly computerized avatars of Vehicular Metaverse Users (VMUs), the Vehicle Twins (VTs) deployed in edge servers can provide valuable metaverse services to improve driving safety and on-board satisfaction for their VMUs throughout journeys. To maintain uninterrupted metaverse experiences, VTs must be migrated among edge servers following the movements of vehicles. This can raise concerns about privacy breaches during the dynamic communications among vehicular edge metaverses. To address these concerns and safeguard location privacy, pseudonyms as temporary identifiers can be leveraged by both VMUs and VTs to realize anonymous communications in the physical space and virtual spaces. However, existing pseudonym management methods fall short in meeting the extensive pseudonym demands in vehicular edge metaverses, thus dramatically diminishing the performance of privacy preservation. To this end, we present a cross-metaverse empowered dual pseudonym management framework. We utilize cross-chain technology to enhance management efficiency and data security for pseudonyms. Furthermore, we propose a metric to assess the privacy level and employ a Multi-Agent Deep Reinforcement Learning (MADRL) approach to obtain an optimal pseudonym generating strategy. Numerical results demonstrate that our proposed schemes are high-efficiency and cost-effective, showcasing their promising applications in vehicular edge metaverses.
\end{abstract}


\IEEEpeerreviewmaketitle

\section{Introduction}
The fast evolution of Internet of Things (IoT) systems has paved the way for the novel paradigm of metaverse, which is considered a creative application of Beyond 5G (B5G) networks to meet people's growing demands for hyper spatio-temporal and surreal digital services in the near future\cite{wang2022metaversesurvey,xu2023edge}. By integrating the technologies of edge computing and intelligent transportation systems, metaverses can further transition into a distinct paradigm called vehicular edge metaverses~\cite{liu2023reputation}. Functioning as surreal realms that merge virtual spaces with the physical space at the network edge, vehicular edge metaverses can offer a range of remarkable metaverse services such as Augmented Reality (AR) navigation, to Vehicular Metaverse Users (VMUs) (i.e., drivers and passengers within vehicles) with lower latency and higher fidelity~\cite{zhang2023learning}. These services can significantly increase VMUs' driving safety and on-board satisfaction throughout their journey\cite{luo2023privacy}. Vehicle Twins (VTs) as specific Al agents\cite{lifeifeicvpr2024}, are one of the core components of delivering metaverse services, which cover the entire life cycles of the vehicle and VMUs in vehicular edge metaverses. Embedded with versatile multimodal Large Language Models (LLMs)~\cite{cuiWACV2024}, the VTs can process multimodal sensory inputs from their VMUs (e.g., gestures and speech) and the vehicle (e.g., sensing data collected by LiDAR and cameras) to enhance environmental perception and understanding. The VTs in virtual spaces are capable of continuously updating themselves through interacting with other VTs and their associated VMUs, thereby providing customized services back to the VMUs in the physical space~\cite{zhang2023learning}.

Owing to the inherent limitations in computing and storage resources of vehicles, the memory and computation-intensive tasks of maintaining VTs and running LLMs should be offloaded to edge servers along the roadside, such as base stations and RoadSide Units (RSUs)\cite{zhang2023learning}. However, given the restricted communication coverage of edge servers, VTs often necessitate migrations among edge servers along the route of their associated VMUs to provide uninterrupted metaverse services~\cite{luo2023privacy}. Therefore, a substantial volume of communication takes place within vehicular edge metaverses. On one hand, vehicles with VMUs must broadcast safety messages including their current location information to nearby vehicles and edge servers, thereby enhancing mutual awareness of surrounding traffic conditions and improving driving safety~\cite{petit2014pseudonym,kang2017p3}. On the other hand, VTs deployed in edge servers should connect with their VMUs in the physical space for real-time data synchronization and metaverse service provisioning, and interact with other VTs in virtual spaces for global information acquisition during VT migrations\cite{zhang2023learning,luo2023privacy}. These communication processes pose a potential risk of privacy leakages\cite{wang2023DTsurvey}, as malicious attackers could exploit the background information (e.g., location) behind the messages to infer sensitive data, establishing mapping relationships between the identities of VMUs and VTs for constant tracking\cite{luo2023privacy}. Fortunately, as temporary authorized identifiers issued by trusted entities, pseudonyms offer a credible solution for identity anonymization by hiding the true identities of both VMUs and VTs\cite{xu2021efficient}. Through synchronous pseudonym changes, the VMUs and the VTs can increase their privacy levels collectively\cite{luo2023privacy}.

    Although employing pseudonyms can improve privacy protection, incorporating the pseudonym scheme into vehicular edge metaverses remains several challenges that must be addressed before practical implementation. These challenges include: 1) On account of the large-scale use of VMU and VT pseudonyms in vehicular edge metaverses, the issue of pseudonym management becomes considerably more intractable. Traditionally, the Trusted Authority (TA) in the cloud layer is responsible for generating, distributing, and revoking pseudonyms throughout the vehicular networks\cite{petit2014pseudonym}. However, this approach may incur an unprecedentedly overwhelming management overhead in the vehicular edge metaverses. 2) Since pseudonyms are typically stored in centralized storage devices within vehicular networks, the vulnerability of sensitive identity privacy to external violations is heightened. Unauthorized access by malicious attackers to these devices could result in the exposure of all identity information in the metaverse system, leading to severe privacy breaches\cite{wang2023DTsurvey,fang2021lightweight}. 3) To maximize utility, both VMUs and VTs need to know where and when to change pseudonyms is better. However, there still lacks a generalized metric to measure the level of privacy protection after pseudonym changes, which significantly hinders the application of privacy-preserving pseudonym schemes in vehicular edge metaverses.

With the above motivation, we resort to blockchain-based pseudonym management for VT migrations in vehicular edge metaverses in this paper. The major contributions of this paper are summarized as follows:
\begin{itemize}
    \item We design a novel cross-metaverse framework for vehicular edge metaverses, with its hierarchical architecture enabling pseudonym management and metaverse service provisioning in an efficient way. The global metaverse consists of multiple local metaverses collaborating to complete dual pseudonym management, thereby ensuring privacy protection for both VMUs and VTs.
    \item To ensure the pseudonym unforgeability and metaverse robustness, we utilize the cross-chain technology combining the notary mechanism to facilitate decentralized and secure pseudonym distribution and revocation during VT migrations in vehicular edge metaverses.
    \item Furthermore, we propose a new metric named Degree of Privacy Entropy (DoPE), to quantify the level of privacy protection after pseudonym changes for the VMUs and VTs. Based on DoPE and inventory theory, we formulate the optimization problem of pseudonym generation within the entire metaverse system.
    \item Given the variability of pseudonym demands as multiple VMUs and VTs dynamically request pseudonyms within different local metaverses, we employ an MADRL algorithm to derive the optimal pseudonym generating strategy in vehicular edge metaverses.
\end{itemize}

The rest of this paper is organized as follows. We first review the related literature in Section \ref{Related Work}. In Section \ref{System Model}, we examine the components and security requirements in vehicular edge metaverses. Then, the details of our cross-metaverse framework are presented in Section \ref{V2DPS}. Following that, the problem of pseudonym generation incorporating the DoPE metric and inventory theory is formulated in Section \ref{Problem formulation and Solution}. To address the problem, we leverage an MADRL algorithm based on edge learning technology in Section \ref{Edge_MADRL}. The performance evaluations including security analyses and numerical results of our proposed framework are performed in Section \ref{Performance Evaluation}. Finally, Section \ref{Conclusion} concludes this paper.

\section{Related Work}
\label{Related Work}
\subsection{Pseudonym Management Framework}
Some research works have been conducted to explore pseudonym management in Internet of Vehicles (IoVs). For instance, the authors in \cite{khan2022privacy} proposed a privacy-preserving identity management scheme for vehicular social networking to enhance the security of vehicles by logging and monitoring malicious pseudonyms. However, their centralized architecture incurs significant communication overhead, posing a challenge for managing both VMU and VT pseudonyms in vehicular edge metaverses. To address this, the authors in \cite{kang2017p3} presented a three-layer (cloud-fog-user) architecture for pseudonym management, in which they harnessed pseudonym fogs dispersed at the network edge to reduce management overhead. Nevertheless, their approaches may lead to sensitive data disclosure or tampering, as the pseudonym fogs are vulnerable to external attackers. To resolve the problem of single point of failure, the authors in \cite{cheng2023conditional} further proposed a blockchain-assisted pseudonym management scheme for multi-domain IoVs, which involves a blockchain network jointly maintained by TA and key generation authority to store pseudonym identities and status for vehicles. However, if the single blockchain network under their scheme fails, all identity privacy in the metaverses could still be divulged.

\subsection{Cross-chain for Metaverse}
As an effective tool for data security, the cross-chain technology shows its unique advantages in enhancing the interoperability and scalability of blockchain networks\cite{huang2021survey}. For metaverse applications, the authors in \cite{wang2022metaversesurvey} proposed to use cross-chain technology to assist transaction authentication across sub-metaverses, but they did not conduct experiments to demonstrate the practicability of their idea. The authors in \cite{kangjinbo2023healthcare} presented a cross-chain empowered federated learning framework for secure data training in healthcare metaverses. More recently, the authors in \cite{li2023metaopera} proposed a cross-metaverse protocol, with the cross-chain technology serving as a key component to achieve metaverse interoperability. Supported by their proposed protocol, the users within different metaverses are capable of quickly interacting with each other. Although previous works have delved into the applications of cross-chain technology in metaverses, none of them explores its significant potential for pseudonym management, particularly in vehicular edge metaverses.

\subsection{Privacy Metric}
To better assess the effectiveness of privacy protection of pseudonym schemes, researchers have proposed various privacy metrics. For instance, the authors in \cite{liu2018location} conducted a systematic study on location privacy, which is defined by three attributes of vehicle location information, namely, identity, position, and time. However, they did not provide a mathematical definition to quantify privacy levels. Building on this, the authors in \cite{kang2017p3} utilized a metric called privacy entropy to measure the uncertainty of mapping pseudonyms to real identities from the perspective of adversaries. Nevertheless, this metric fails to capture the impact of each pseudonym change on location privacy. Therefore, the authors in \cite{liu2019uncoordinated} utilized a metric named pseudonym age to measure privacy levels, defined as the time interval between the last pseudonym change and the current one. However, this metric solely takes the time dimension into account but ignores the strength factor. As each pseudonym change yields different degrees of privacy enhancement based on real-time traffic conditions\cite{kang2017p3}, the pseudonym age does not reflect the location privacy well. Consequently, there is still a need for a generalized metric that can analytically quantify the degree of privacy protection after pseudonym changes, particularly in the context of pseudonym-based vehicular edge metaverses.

\subsection{Resource Optimization for Pseudonym Management}
For the sake of cost-effective privacy preservation at the network edge, researchers have investigated optimization problems with regard to pseudonym management. For example, the authors in \cite{artail2015pseudonym} developed an optimization algorithm to solve the pseudonym shuffling problem among roadside units in a distributed manner. Additionally, the authors in \cite{chaudhary2019pseudonym} leveraged the genetic algorithm to execute pseudonym generation, targeting the problem of location privacy preservation in vehicular ad hoc networks. Recently, the authors in \cite{luo2023privacy} conducted a case study on pseudonym distribution formulated by the VMU utility with inventory theory. However, most existing research overlooks the intricate scenario of multi-provider multi-consumer pseudonym management, thereby constraining their performances in vehicular edge metaverses. Moreover, these studies typically rely on heuristics to solve the formulated problem, which lacks practicality in real-world applications. Therefore, it is urgent to develop a practical learning-based algorithm that can obtain an optimal pseudonym generating strategy in vehicular edge metaverses.

\section{System Model}
\label{System Model}
\subsection{Network Model in Vehicular Edge Metaverses}
\begin{itemize}

    \item \textbf{Vehicle Twins (VTs):} The multimodal LLM-based VTs are one kind of powerful AI agents, which can evolve through engaging online in the vehicular edge metaverse, namely, synchronizing with their corresponding VMUs and interacting with other VTs\cite{luo2023privacy}. In this way, VTs can better perform complex tasks by Chain-of-Thought (CoT) reasoning and planning in the edge layer, making more accurate decisions to enhance VMUs’ driving safety and enrich on-board experience\cite{lifeifeicvpr2024,cuiWACV2024}. However, the VT migrations may expose sensitive data (e.g., identity privacy) during dynamic communications within vehicular edge metaverses\cite{wang2023DTsurvey}. To mitigate potential privacy breaches, VTs can employ pseudonyms to effectively mask their true identities for anonymous communications\cite{luo2023privacy}.

    \item \textbf{Vehicular Metaverse Users (VMUs):} VMUs connect to edge servers to access vehicular edge metaverses through portable immersive devices, such as Head-Mounted Displays (HMDs). To enjoy personalized services, the VMUs should continuously upload real-time sensing data collected by vehicular sensors to update their VTs\cite{zhang2023learning}. In addition, their vehicles need to periodically broadcast safety messages to increase the contextual awareness of surrounding traffic conditions\cite{kang2017p3}. By synchronously changing pseudonyms with their VTs, VMUs can evade the continuous tracking by attackers during data synchronization and safety message broadcasting, thus safeguarding location privacy throughout the journeys\cite{luo2023privacy}.

    \item \textbf{Edge Servers:} Vehicular edge metaverses are generally maintained by numerous edge servers geographically adjacent to VMUs to reduce service provisioning delays. With adequate computation, communication, and storage resources, the edge servers can perform latency-sensitive and computationally complicated tasks, such as VT simulation and visualization renderings\cite{zhang2023learning}. Moreover, the edge servers expedite pseudonym management including pseudonym issuance, storage, allocation, and revocation\cite{petit2014pseudonym,kang2017p3}. Combined with the blockchain technology, the edge servers function as miners to participate in consensus related to pseudonyms\cite{kangjinbo2023healthcare}. This integration fosters pseudonym security in vehicular edge metaverses.
    
    \item \textbf{Trusted Authority (TA):} The TA can be a government agency equipped with anti-attack hardware to effectively thwart network attacks\cite{kang2017p3}. Located in the cloud layer, the TA with ample computing resources plays a pivotal role in overseeing pseudonym management operations as well as supervising VMUs, VTs, and edge servers within the entire metaverse. By synchronizing pseudonym information from edge servers on the blockchain network, the TA can verify pseudonym identities swiftly upon receiving messages related to pseudonyms, such as pseudonym renewals or misbehavior reports\cite{cheng2023conditional}. The pseudonym management is dominated by the fully trusted TA, thereby maintaining the pseudonym sustainability and accountability in vehicular edge metaverses\cite{luo2023privacy}.

\end{itemize}

\subsection{Security Requirements}
The security threats in vehicular edge metaverses consist of purposeful behaviors and fatal attacks by malicious attackers. The attackers can be categorized into three main types: 1) Compromised VMUs, which eavesdrop basic safety messages from legitimate VMUs\cite{kang2017p3}; 2) Malicious VTs, which engage in intentional interactions with legal VTs to steal their sensitive privacy or disseminate fake news\cite{luo2023privacy}; 3) Semi-trusted edge servers, which are curious about the data synchronization between legal VMUs and VTs\cite{li2014acpn,wang2023DTsurvey}. These edge servers are susceptible to hijacking by external attackers, affecting the normal pseudonym management in vehicular edge metaverses.

\begin{figure*}[t]
\vspace{-0.5cm}
\centering{\includegraphics[width=0.8\textwidth]{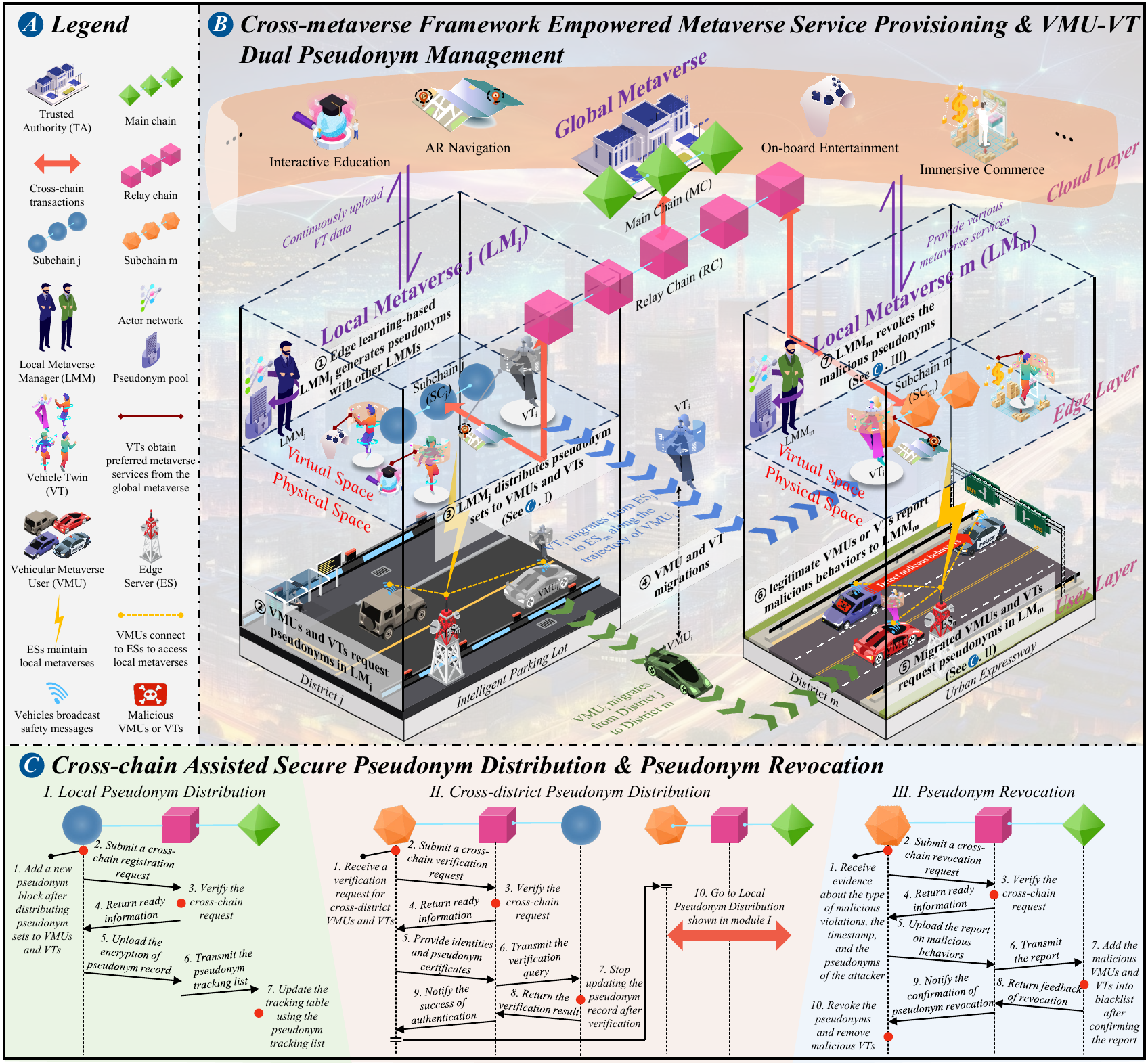}}
\caption{The cross-metaverse paradigm integrated with cross-chain technology for twin migrations in vehicular edge metaverses. \textit{Part. A} presents the element description in the vehicular edge metaverses. \textit{Part. B} depicts the cross-metaverse framework in detail. \textit{Part. C} provides the details of cross-chain transactions for pseudonym distribution and revocation in \textit{Part. B}.}\label{System figure}
\end{figure*}
Considering the existence of the aforementioned attackers, to ensure normal operation in vehicular edge metaverses, the following security requirements must be strictly satisfied:
\begin{itemize}
    \item \textbf{\emph{Anonymity:}} Anonymity is the foundation of vehicular edge metaverses, ensuring that the real identities of VMUs and their associated VTs remain undisclosed to other entities within the metaverse.
    \item \textbf{\emph{Unlinkability:}} Unlinkability ensures that colluding attackers cannot link the identities of VMUs with their associated VTs during vehicle and VT migrations, thus protecting location privacy for both VMUs and VTs.
    \item \textbf{\emph{Immutability:}} Immutability means that the sensitive pseudonym information should be guaranteed for data integrity and cannot be easily obtained or modified by external attackers.
    \item \textbf{\emph{Conditional traceability:}} Conditional traceability grants the exclusive right to trace the true identities of VMUs and VTs and revoke the pseudonym use of malicious entities only to the creditable entities.
    \item \textbf{\emph{Robustness:}} Robustness is the ability of the blockchain network in vehicular edge metaverses to withstand external attacks and ensure on-chain data protection to a certain extent.
    \item \textbf{\emph{Efficiency:}} Due to ubiquitous communications in vehicular edge metaverses, the pseudonym management scheme should be highly efficient and cost-effective to ensure affordable and sustainable privacy protection.
\end{itemize}

\section{Cross-metaverse Empowered VMU-VT Dual Pseudonym Management Framework}
\label{V2DPS} 
\subsection{Overview of Hierarchical and Decentralized Architecture} 
To satisfy the security requirements mentioned above, we design a cross-metaverse empowered VMU-VT dual pseudonym management framework, as illustrated in Fig. \ref{System figure}. The vehicular edge metaverse is regarded as a global metaverse encompassing multiple local metaverses to provide a variety of stunning metaverse services for VTs at the network edge, such as AR navigation and on-board entertainment\cite{jiayi2023attack,li2023filling}. Afterwards, the associated VMUs in the user layer can obtain global traffic information while experiencing various kinds of metaverse services. This hierarchical architecture of our proposed framework expedites the metaverse service provisioning via cross-metaverse interactions. On one hand, the VTs within local metaverses continuously upload real-time data to the global metaverse, enabling it to render and offer customized services back to local metaverses. On the other hand, cross-metaverse interactions contribute to overcoming spatial limitations among different physical districts. For instance, if a traffic accident occurs in a certain local metaverse, the VMU witnessing the accident can notify its VT, and then the VT will forward this notification to the VTs in other local metaverses. In this way, the VMUs receiving the notification from their associated VTs can choose alternative routes, thus preventing the accident from escalating further.

In addition to facilitating metaverse service provisioning, the proposed cross-metaverse hierarchical architecture also plays a vital role in pseudonym management. Based on the above contents, we assume that the global metaverse is equipped with a TA in the cloud layer, and each local metaverse is equipped with a Local Metaverse Manager (LMM) in the edge layer. Serving as regional trusted authorities, the LMMs are qualified to carry out the tasks of pseudonym management traditionally handled by the TA\cite{kang2017p3}. For instance, LMMs are capable of promptly generating and distributing pseudonyms to VMUs and VTs within local metaverses, substantially relieving the burdens of issuing, storing, transmitting, and recording pseudonyms in the cloud center.

Moreover, leveraging the blockchain technology to record pseudonym identities can ensure pseudonym confidentiality. However, in a single-chain system, the limited throughput and single-chain architecture become a constraint when handling massive pseudonymous transactions, potentially degrading the performance of the blockchain and causing irreversible damage to the metaverse\cite{kangjinbo2023healthcare}. Fortunately, the application of cross-chain technology can address the above challenges. As shown in Fig. \ref{System figure}, the hierarchical decentralized cross-chain architecture consists of a main chain, a relay chain, and multiple subchains. The main chain is maintained by the fully trusted nodes in the cloud layer (e.g., the TA), responsible for recording and verifying global pseudonyms and reports from various local metaverses. The relay chain facilitates the transmission of ciphertext, namely cross-chain requests between the main chain and subchains. Each local metaverse maintains a subchain in the edge layer, where edge servers function as miners using distributed consensus to add pseudonym blocks and LMMs serve as notaries to verify these blocks for cross-chain transactions\cite{du2021resource}.

In summary, the hierarchical cross-metaverse architecture enables efficient pseudonym management while the decentralized cross-chain architecture ensures sensitive privacy isolation. Consequently, our proposed framework effectively improves both management efficiency and data security for pseudonyms in vehicular edge metaverses.
\begin{table}[t]\label{symbol}
	\renewcommand{\arraystretch}{1.3} 
	\caption{Key symbols used in this paper.}\label{table} \centering 
	\begin{tabular}{m{0.14\textwidth}<{\centering}|m{0.3\textwidth}<{\raggedright}} 
		\Xhline{1px}		
		\textbf{Notation} & \textbf{Definition}\\	
		\Xhline{1px}
            $PID_{VMU_i}^k$ & The $k^{\mathrm{th}}$ pseudonym of $VMU_i$. $VMU_i$ requests a pseudonym set with $w$ pseudonyms, ${\{PID_{VMU_i}^k\}}_{k=1}^w={\{PID_{VMU_i}^k\}}$ \\
            \hline
            $PID_{VT_i}^l$ & The $l^{\mathrm{th}}$ pseudonym of $VT_i$. $VT_i$ requests a pseudonym set with $u$ pseudonyms, ${\{PID_{VT_i}^l\}}_{l=1}^u={\{PID_{VT_i}^l\}}$ \\
            \hline
            $PK_{VMU_i},SK_{VMU_i}$ & Public and private key pair of $VMU_i$ \\
            \hline
            $PK_{VT_i},SK_{VT_i}$ & Public and private key pair of $VT_i$\\
            \hline
            $\{x\} $  & A set with element $x$ \\
            \hline
            $i \mathop \to \limits^{(j)} k $  & Entity $i$ sends a message to entity $k$ (through entity $j$)\\
            \hline
            $i \stackrel{(n)}{\Longrightarrow}c $  & Entity $i$ adds a new block to blockchain $c$ (after the authentication of notary $n$) \\
            \hline
            $x||y$ &  Element $x$ concatenates to $y$  \\
            \hline
            ${E_{P{K_i}}}(m)$ & Encryption of message $m$ using the public key of entity $i$\\
            \hline
            $Ts$ &  Timestamp of event\\
            \Xhline{1px}
	\end{tabular}\label{symbol}
\vspace{-0.5cm}
\end{table}

\subsection{Cross-chain Assisted Secure Pseudonym Management}
Here, we provide a detailed description of cross-chain assisted pseudonym management processes in vehicular edge metaverses. For convenience, we list the symbols used in our proposed framework in Table \ref{symbol}.

We adopt a lightweight Boneh-Boyen short signature scheme \cite{boneh2004short} for initial startup and key generation. When $VMU_i$ with its true identity $ID_{VMU_i}$ first joins the $j^{\mathrm{th}}$ local metaverse, it sends an initial metaverse registration request to the nearest edge server. Then, the edge server will create $VT_i$ for $VMU_i$ after verifying $ID_{VMU_i}$, and both $VMU_i$ and $VT_i$ obtain its public\//privacy key pairs and corresponding certificates (denoted as $PK_{VMU_i}, PK_{VT_i}, SK_{VMU_i}, SK_{VT_i}, Cert_{VMU_i}$, and $Cert_{VT_i}$) from the TA\cite{kang2017p3}. Afterwards, The TA notifies $LMM_j$ to distribute a set of pseudonyms $\{PID_{VMU_i}^k\}_{k=1}^w$ and $\{PID_{VT_i}^l\}_{l=1}^u$ to $VMU_i$ and $VT_i$, respectively. These pseudonyms are attached with the corresponding public\//private key pairs and certificates (denoted as $PID_{VMU_i}^k, PID_{VT_i}^l, SK_{PID_{VMU_i}^k}, SK_{PID_{VT_i}^l}, Cert_{PID_{VMU_i}^k}$, and $Cert_{PID_{VT_i}^l}$). Following that, the TA creates a tracking table on the main chain $MC$, logging the true and pseudonym identities of $VMU_i$ and $VT_i$ as well as the pseudonym issuer $LMM_j$\cite{kang2017p3}. Meanwhile, the edge server adds a pseudonym registration block onto the subchain $SC_j$. Finally, the TA allocates a tracking list \{$PK_{VMU_i}$, $Cert_{VMU_i}$, $LMM_j$, $\{PID_{VMU_i}^k\}$, $\{Cert_{PID_{VMU_i}^k}\}$\} to all local metaverses after encryption using LMMs' public keys. Here, we consider that all the LMMs are fully trusted while edge servers are semi-trusted\cite{li2014acpn}. The edge servers scattered across local metaverses link with each other via wired communications for better cooperation and mutual supervision\cite{kang2017p3}. The LMMs also supervise these edge servers and ban them from linking to the subchains for a specific period in case of detecting their misbehaviors during pseudonym management\cite{khan2022privacy}.

\begin{figure}[t]
\vspace{-0.5cm}
\small
\noindent\rule{0.49\textwidth}{1pt}
\vspace{-5pt} 
\leftline{\textbf{Protocol 1: Basic Operations of Pseudonym Management}}\\
\label{protocol1}
\noindent\rule{0.49\textwidth}{1pt}
\begin{small}
\noindent  \textbf{1.} $LMM_j$: generate $G_{j}^t$ pseudonyms via the actor network trained \\
\indent \qquad \qquad \quad by \textbf{Algorithm} \ref{algorithm} and store them in the pseudonym pool \\
\noindent  \textbf{2.} $VMU_i$: broadcast safety messages with VMU pseudonym \\
\indent \qquad \qquad \quad $PID_{VMU_i}^k$ to neighboring VMUs and nearby $ES_m$ \\
\noindent  \textbf{3.} $VT_i$: interact with other VTs in virtual spaces and connect to\\
\indent \qquad \quad \enspace $VMU_i$ with VT pseudonym $PID_{VT_i}^l$ \\
\noindent  \textbf{4.} $VMU_i \enspace \& \enspace VT_i$: change pseudonym synchronously \\
\noindent \textbf{5.} \textbf{if} $VMU_i$ is (going to run out of pseudonyms issued by $LMM_j$)\\
\indent \quad \noindent \textbf{5.1} \textbf{if} $VMU_i$ is (within the coverage of the $j^{\mathrm{th}}$ local metaverse)\\
\indent \enspace \qquad~ \noindent $VMU_i$: Go to \textbf{Protocol 2}\\
\indent \quad \quad \enspace \textbf{else} \\
\indent \enspace \qquad~ \noindent $VMU_i$: Go to \textbf{Protocol 3}\\
\indent \quad \quad \enspace \textbf{endif} \\
\indent \quad \noindent \textbf{endif}\\
\noindent  \textbf{6.} $LMM_j$: Calculate the total pseudonym demand $D_{j}^t$ within the\\
\vspace{-5pt} 
\indent \qquad \qquad \quad $j^{\mathrm{th}}$ local metaverse\\
\end{small}
\noindent\rule{0.49\textwidth}{1pt}
\vspace{-0.5cm}
\end{figure}

\subsubsection{Basic operation}
In our proposed scheme, the distributed LMMs are trained in parallel based on edge learning technology\cite{xu2023edge} to periodically generate pseudonyms and store them in pseudonym pools for subsequent allocation within local metaverses (see step $\textcircled{1}$ in Fig. \ref{System figure}). When $VMU_i$ moves within the $j^{\mathrm{th}}$ local metaverse, the vehicle broadcasts safety messages with pseudonym $\{PID_{VMU_i}^k \}_{k=1}^w$ every 300 ms~\cite{kang2017p3}. For simplicity, here we assume that both $VMU_i$ and $VT_i$ request the same number of pseudonyms\cite{xu2021efficient}. $VMU_i$ and $VT_i$ can synchronously change their pseudonyms to prevent continuous tracking by attackers\cite{luo2023privacy}. To ensure the validity of safety messages, $VMU_i$ signs its messages with a timestamp to guarantee message freshness, in which pseudonym certificates are attached for identity verification~\cite{kang2017p3}. Before depleting all available pseudonyms, $VMU_i$ and $VT_i$ request new pseudonyms from the local metaverse where they reside. Then, the local or cross-district pseudonym distribution protocol is correspondingly executed according to the $VMU_i$'s current position. Further details can be found in \textbf{Protocol 1}.

\begin{figure}[t]
\vspace{-0.5cm}
\small
\noindent\rule{0.49\textwidth}{1pt}
\vspace{-5pt} 
\leftline{\textbf{Protocol 2: Local Pseudonym Distribution}}\\
\label{protocol2}
\noindent\rule{0.49\textwidth}{1pt}
\begin{small}
\noindent  \textbf{1.} ${VMU_i}\mathop\to \limits^{ES_m} LMM_j$: \\
\indent\quad \quad $request_{\_VMU}= {E_{PK_{LMM_j}}}(Pseu\_request||PK_{VMU_i}||$\\ 
\indent\quad~\qquad~\qquad~~\qquad~~ $PID_{VMU_i}^n||Cert_{VMU_i}||Cert_{PID_{VMU_i}^n})$,\\
\indent \quad~ where $Pseu\_request = \{locatio{n_i}||D_{j,i}^t||Ts\}$\\
\noindent  \textbf{2.} $VT_i \to LMM_j$:\\
\indent\quad \quad $request_{\_VT}= {E_{PK_{LMM_j}}}(Pseu\_request||PK_{VT_i}||PID_{VT_i}^n||$\\ 
\indent\quad~\qquad~\qquad~\qquad~  $Cert_{VT_i}||Cert_{PID_{VT_i}^n})$\\
\noindent \textbf{3.} $LMM_j$: decrypt $request_{\_VMU}$ and $request_{\_VT}$ with $SK_{LMM_j}$\\
\indent \qquad \qquad \quad to obtain $VMU_i$ and $VT_i$'s identities for verification\\
\noindent \textbf{4.} \textbf{if} $LMM_j$ verified ($PID_{VMU_i}^n$, $PID_{VT_i}^n$, $PK_{VMU_i}$ and $PK_{VT_i}$\\
\indent \quad \quad are on $SC_j$) and ($VMU_i$ is within $ES_m$)\\
\indent \quad~ \noindent \textbf{4.1} ${LMM_j}\mathop\to \limits^{ES_m} VMU_i$: \\
\indent \quad \qquad~ $Reply_{\_VMU} = {E_{P{K_{PID_{VMU_i}^n}}}}(\{ PID_{VMU_i}^k,S{K_{PID_{VMU_i}^k}},$\\
\indent\quad~\qquad~\qquad~\qquad\qquad $Cer{t_{PID_{VMU_i}^k}}\} _{k = 1}^w)||Ts$\\
\indent \quad~ \noindent \textbf{4.2} ${LMM_j}\mathop\to \limits VT_i$: \\
\indent \quad \qquad~ $Reply_{\_VT} = {E_{P{K_{PID_{VT_i}^n}}}}(\{ PID_{VT_i}^l,S{K_{PID_{VT_i}^l}},$\\
\indent\quad~\qquad~\qquad\qquad\qquad $Cer{t_{PID_{VT_i}^l}}\} _{l = 1}^u)||Ts$\\
\indent \quad~ \noindent \textbf{4.3} $ES_m\stackrel{LMM_j}{\Longrightarrow} SC_j$: \\
\indent \qquad~\quad  $Record = (P{K_{VMU_i}}||\{ PID_{VMU_i}^k,S{K_{PID_{VMU_i}^k}}$,\\
\indent\quad~\qquad~\qquad\qquad\enspace $Cer{t_{PID_{VMU_i}^k}}\}||Cer{t_{VMU_i}})||Ts$\\
\indent \quad~ \noindent \textbf{4.4} $SC_j \mathop\to \limits {RC}$: cross-chain registration request\\
\indent \quad~ \noindent \textbf{4.5} $RC \mathop\to \limits {SC_j}$: ready information after authenticating $SC_j$\\
\indent \quad~ \noindent \textbf{4.6} ${SC_j} \mathop\to \limits ^{RC} {MC}$: \\
\indent \quad~~ \quad~ $Tracking\_list = {E_{PK{}_{TA}}}(Record||Cer{t_{LM{M_j}}})||Ts$\\
\indent \quad~ \noindent \textbf{4.7} TA: download $Tracking\_list$ from $MC$ and decrypt with \\
\indent \qquad \qquad \enspace~  $SK_{TA}$ for information update\\
\indent \quad~ \noindent \textbf{4.8} ${TA}\Longrightarrow {MC}$: \\
\indent \quad~~ \quad~ $Tracking\_table = \{ID_{VMU_i}||PK_{VMU_i}||PK_{VT_i}||$\\
\indent \qquad\qquad\qquad\qquad~\qquad~\quad\enspace $Cert_{VMU_i}||Cert_{VT_i}||Cert_{LMM_j}||$\\
\indent \qquad\qquad\qquad\qquad~\qquad~\quad\enspace $Ts||\{PID_{VMU_i}^k\}||\{PID_{VT_i}^l\}||$\\
\indent \qquad\qquad\qquad\qquad\qquad\quad\quad $\{Cer{t_{PID_{VMU_i}^k}}\}||\{Cer{t_{PID_{VT_i}^l}}\}\}$\\
\indent \quad\quad  \textbf{else} \\
\indent \qquad~  \quad $LMM_j$: do not reply\\
\vspace{-5pt} 
\indent \quad~ \textbf{endif}\\
\end{small}
\noindent\rule{0.49\textwidth}{1pt}
\vspace{-0.5cm}
\end{figure}
\subsubsection{Local pseudonym distribution}
Before exhausting all pseudonyms, $VMU_i$ and $VT_i$ request new pseudonyms from the nearest edge server (denoted as $ES_m$) in the $j^{\mathrm{th}}$ Local Metaverse $(LM_j)$ where they previously conducted initialization (see step $\textcircled{2}$ in Fig. \ref{System figure}). The request includes the number of requesting pseudonyms, the current location, the public key, the pseudonym being used, and corresponding certificates, all encrypted with $LMM_j$'s public key\cite{kang2017p3}. After confirming that $VMU_i$ and $VT_i$ are in $LM_j$ while verifying their identities on $SC_j$, $LMM_j$ distributes the pseudonym set $\{PID_{VMU_i}^k\}$ and $\{PID_{VMU_i}^l\}$ to $VMU_i$ and $VT_i$, respectively (see step $\textcircled{3}$ in Fig. \ref{System figure}). Subsequently, $ES_m$ generates a pseudonym registration block, which is added to $SC_j$ with a distributed consensus algorithm\cite{cheng2023conditional}. Thereafter, $SC_j$ submits a cross-chain registration request to the relay chain $RC$. After verifying by both $LMM_j$ and $RC$, the block will finally forward to the main chain (\textit{Module} $\Rmnum{1}$ of \textit{Part. C} in Fig. \ref{System figure}). More details are presented in \textbf{Protocol 2}.
\begin{figure}[t]
\vspace{-0.5cm}
\small
\noindent\rule{0.49\textwidth}{1pt}
\vspace{-5pt} 
\leftline{\textbf{Protocol 3: Cross-district Pseudonym Distribution}}\\
\label{protocol3}
\noindent\rule{0.49\textwidth}{1pt}
\begin{small}
\noindent  \textbf{1.} ${VMU_i}\mathop\to \limits^{ES_{n}} LMM_m$: \\
\indent\quad \quad $request_{\_VMU}= {E_{PK_{LMM_m}}}(Pseu\_request||PK_{VMU_i}||$\\ 
\indent\quad~\qquad~\qquad~~\qquad~~ $PID_{VMU_i}^{w-1}||Cert_{VMU_i}||Cert_{PID_{VMU_i}^{w-1}})$\\
\noindent  \textbf{2.} $VT_i \to LMM_m$:\\
\indent\quad \quad $request_{\_VT}= {E_{PK_{LMM_m}}}(Pseu\_request||PK_{VT_i}||PID_{VT_i}^{u-1}||$\\ 
\indent\quad~\qquad~\qquad~\qquad~ $Cert_{VT_i}||Cert_{PID_{VT_i}^{u-1}})$\\
\noindent  \textbf{3.} $LMM_m$: decrypt $request_{\_VMU}$ and $request_{\_VT}$ with $SK_{LMM_j}$\\
\indent \qquad \qquad \quad \! to obtain $VMU_i$ and $VT_i$'s identities\\
\noindent  \textbf{4.} \textbf{if} $LMM_m$ verified ($PID_{VMU_i}^n$, $PID_{VT_i}^n$, $PK_{VMU_i}$ and\\
\indent \quad \quad $PK_{VT_i}$ are \textbf{not} on $SC_m$) \\
\indent \quad~ \noindent \textbf{4.1} ${SC_m}\mathop\to \limits{RC}$: cross-chain verification request\\
\indent \quad~ \noindent \textbf{4.2} ${RC}\mathop\to \limits{SC_m}$: ready information after verifying $SC_m$\\
\indent \quad~ \noindent \textbf{4.3} ${SC_m}\mathop\to \limits^{RC} {SC_j}$: \\
\indent \quad~~ \quad~ $Query = {E_{PK_{LMM_j}}} (PID_{VMU_i}^n||PID_{VT_i}^n||PK_{VMU_i}||$\\
\indent \qquad\qquad\quad~~ \quad~\enspace $PK_{VT_i}||Cer{t_{LM{M_m}}})||Ts$ \\
\indent \quad~ \noindent \textbf{4.4} ${LMM_j}$: authenticate whether $VMU_i$ and $VT_i$'s identities\\
\indent \qquad~\qquad \quad \quad \quad are on ${SC_j}$ via $record$\\
\indent \quad~ \noindent \textbf{4.5} \textbf{if} both $VMU_i$ and $VT_i$'s identities are (verified on ${SC_j})$\\
\indent \qquad \quad~ \enspace $SC_j \mathop\to SC_m$: $identity\enspace authenticated = True$\\
\indent \qquad \quad~ \enspace${LMM_j}$: stop updating the record of $VMU_i$ and $VT_i$ \\
\indent \qquad~ \enspace \textbf{else} \\
\indent \qquad \quad~ \enspace $SC_j \mathop\to SC_m$: $identity\enspace authenticated$ $=$ $False$\\
\indent \qquad \quad~ \enspace ${LMM_j}$: send this abnormal condition to TA \\
\indent \qquad~ \enspace \textbf{endif} \\
\indent \quad~ \noindent \textbf{4.6} \textbf{if} ($identity\enspace authenticated$) \\
\indent \qquad \quad \enspace~\, $LMM_m$: execute pseudonym distribution following the \\
\indent \qquad \qquad \qquad \quad \quad \, processes of \textbf{{4.1} - {4.9}} in \textbf{Protocol 2} \\
\indent \quad\quad\enspace~\,  \textbf{else} \\
\indent \qquad \quad \enspace~\, \textbf{pass}\\
\indent \quad\quad\enspace~\,  \textbf{endif} \\
\indent\quad\enspace \noindent \textbf{else} \\
\indent \qquad~  \quad $LMM_m$: do not reply\\
\vspace{-5pt} 
\indent\enspace\enspace \noindent \textbf{endif}  \\
\end{small}
\noindent\rule{0.49\textwidth}{1pt}
\vspace{-0.5cm}
\end{figure}

\begin{figure}[t]
\vspace{-0.5cm}
\small
\noindent\rule{0.49\textwidth}{1pt}
\vspace{-5pt} 
\leftline{\textbf{Protocol 4: VMU-VT Dual Pseudonym Revocation}}\\
\label{protocol4}
\noindent\rule{0.49\textwidth}{1pt}
\begin{small}
\noindent \textbf{1.} \textbf{if} ${VMU_i}$ detects (${VMU_k}$ misbehaved in the physical space) \\
\indent\quad \quad \textbf{1.1} ${VMU_i}\mathop\to \limits^{ES_n} LMM_m$: \\
\indent\quad \quad \quad~ $report_{VMU} = {E_{P{K_{L{MM_j}}}}}(Cert_{PID_{VMU_i}^n}||message_{VMU})$, \\
\indent\quad \quad \quad~ where $message_{VMU} = \{Cert_{PID_{VMU_k}^m}||type||Ts\}$\\
\indent\quad \noindent \textbf{elif} ${VT_i}$ detects (${VT_k}$ misbehaved in the virtual space)\\
\indent\quad \quad \textbf{1.2} ${VT_i}\mathop\to \limits LMM_m$: \\
\indent\quad \quad \quad~ $report_{VT} = {E_{P{K_{L{MM_j}}}}}(Cert_{PID_{VT_i}^n}||message_{VT})$, \\
\indent\quad \quad \quad~ where $message_{VT} = \{Cert_{PID_{VT_k}^m}||type||Ts\}$\\
\indent\quad \noindent \textbf{endif}  \\
\noindent  \textbf{2.} $LMM_m$: decrypt $report_{VMU}$ or $report_{VT}$ with $SK_{LMM_m}$ \\
\noindent  \textbf{3.} \textbf{if} $LMM_m$ verified $message$ ($message_{VMU}$ or $message_{VT}$) \\
\indent \quad \quad and the identity of reporter ($VMU_i$ or $VT_i$) via $SC_j$ \\
\indent\quad \quad \textbf{3.1} $ES_m\stackrel{LMM_j}{\Longrightarrow} SC_j$: \\
\indent \qquad~\quad\enspace  $report = {E_{PK_{TA}}} (P{K_{VMU_k}}||\{ PID_{VMU_k}^m,S{K_{PID_{VMU_k}^m}}$\\
\indent\quad~\qquad~\qquad\qquad\enspace $Cer{t_{PID_{VMU_k}^m}}\}||Cert_{LMM_m}||Cer{t_{VMU_k}})||$\\
\indent\quad~\qquad~\qquad\qquad\enspace $message||Ts$\\
\indent\quad \quad \textbf{3.2} $SC_j \mathop\to \limits {RC}$: cross-chain revocation request \\
\indent\quad \quad \textbf{3.3} $RC \mathop\to \limits {SC_j}$: ready information after verifying $SC_j$ \\
\indent\quad \quad \textbf{3.4} $SC_j \mathop\to \limits^{RC} {MC}$: $report$ \\
\indent\quad \quad \textbf{3.5} $TA$: decrypt $report$ with $SK_{TA}$ and validate $message$ \\
\indent \qquad \qquad \quad \enspace \, with $Tracking\_table$ on $MC$ \\
\indent\quad \quad \textbf{3.6} \textbf{if} $TA$ confirmed ($VMU_k$ or $VT_k$ misbehaved) \\
\indent \qquad \quad \quad \enspace \noindent $TA$: add $VMU_k$ and $VT_k$ into the blacklist and reveal\\
\indent \qquad \qquad \quad \quad \, the true identity $ID_{VMU_k}$ to all edge servers \\
\indent \qquad \quad \quad \enspace \noindent $LMM_m$: revoke the use of $PID_{VMU_k}$ and $PID_{VT_k}$\\
\indent \qquad \qquad \qquad \quad \quad \, and remove $VT_k$ from $LM_j$\\
\indent\quad \quad \quad\enspace \textbf{else} \\
\indent \qquad \quad \quad~ \noindent $TA$: restrict the right of $VMU_i$ and $VT_i$ to report\\
\indent \qquad \qquad \quad \quad \, violation events \\
\indent\quad \quad \quad\enspace \textbf{endif} \\
\indent\quad \noindent \textbf{else} \\
\indent \qquad~  \quad $LMM_m$: do not reply\\
\vspace{-5pt} 
\indent\enspace \noindent \textbf{endif}  \\
\end{small}
\noindent\rule{0.49\textwidth}{1pt}
\vspace{-0.5cm}
\end{figure}

\subsubsection{Cross-district pseudonym distribution}
Since $VMU_i$ and its vehicle continuously migrate across different districts, the $VT_i$ should be synchronously migrated among edge servers to access different local metaverses\cite{zhang2023learning}. When $VMU_i$ and $VT_i$ have migrated from $LM_j$ to $LM_m$ (see step $\textcircled{4}$ in Fig. \ref{System figure}), they need to reapply for pseudonyms from $ES_n$ using the last pseudonym $PID_{VMU_i}^{w-1}$ and $PID_{VT_i}^{u-1}$ issued by $LMM_j$ (see step $\textcircled{5}$ in Fig. \ref{System figure}). In this case, cross-district pseudonym distribution will be executed. $LMM_m$ audits $VMU_i$ and $VT_i$'s pseudonym identities through cross-chain verification (\textit{Module} $\Rmnum{2}$ of \textit{Part. C} in Fig. \ref{System figure}). After the verification, $LMM_m$ distributes pseudonyms following the processes of local pseudonym distribution. See the details in \textbf{Protocol 3}.

\subsubsection{Dual pseudonym revocation}
In the vehicular edge metaverses, both VMUs and VTs should mutually supervise their neighbors. Legitimate entities can accuse neighboring malicious individuals of engaging in misbehaviors (see step $\textcircled{6}$ in Fig. \ref{System figure}). For example, if the compromised $VMU_k$ with $PID_{VMU_k}^m$ is perpetrating misbehaviors (e.g., spread fake traffic conditions) and is detected by $VMU_i$, $VMU_i$ will record the pertinent violation information including the misbehavior type and timestamp, and report it to $LMM_m$ via $ES_n$. Upon receiving the report, $LMM_m$ first checks the validity of the report as well as the identity of $VMU_i$, and then $ES_n$ adds the report onto the subchain $SC_m$. Then, $SC_m$ will submit a cross-chain pseudonym revocation request (\textit{Module} $\Rmnum{3}$ of \textit{Part. C} in Fig. \ref{System figure}). By auditing the complete identity information on the main chain $MC$, the TA can rapidly check the identities of entities involved in the report and validate the authenticity of the report~\cite{liang2019efficient}. If the misbehaviors are confirmed, TA will reveal the true identity of $VMU_k$ to all edge servers in the metaverse, and then the $LMM_m$ will revoke the use of both VMU and VT pseudonyms and remove $VT_k$ from $LM_m$ at once (see step $\textcircled{7}$ in Fig. \ref{System figure})~\cite{kang2017p3}. Eventually, The malicious $VMU_k$ and $VT_k$ are added to the blacklist, thereby banning them from communicating with other legal entities in vehicular edge metaverses. More details are described in \textbf{Protocol 4}.

\section{Problem Formulation}
\label{Problem formulation and Solution}

\subsection{Privacy Metric and VMU Utility}
VMUs and VTs spread over local metaverses can synchronously change pseudonyms in groups with other entities to jointly increase their privacy levels\cite{luo2023privacy}. Therefore, based on the definition of Age of Information (AoI)\cite{kaul2012real,kosta2017age} used to characterize the latency in status updates, we propose a metric named Degree of Privacy Entropy (DoPE) to quantify location privacy levels for VMUs and VTs after pseudonym changes. Referring to \cite{kang2017p3}, after each pseudonym change at time $t_n^*$, the DoPE can increase to privacy entropy, calculated by
\begin{equation}
    H_n=-\log_2p_i,\quad p_i\in\left[{a}, {b}\right].
\end{equation}
Here, the continuous random variable $p_i$ represents the attackers’ successful tracking probability of $VMU_i$ after its pseudonym changes\cite{kang2017p3}, and $a$ and $b$ denote the reciprocal of the maximum and minimum number of vehicles at a social hotspot\cite{kang2017p3}, respectively.

Without loss of generality, we consider that the autocorrelation of the pseudonym change processes is small because VMUs and VTs do not want their patterns of changing pseudonyms to be discovered by attackers. Therefore, an exponential DoPE is recommended to cope with this scenario\cite{kosta2017age}. As depicted in Fig. \ref{privacy_entropy}, the DoPE $H(t)$ declines exponentially over time, while increasing instantaneously when a VMU or VT replaces its current pseudonym with a new one at $t_n^*$. Note that $H(t)$ passes through the point $(t_{i-1},H_n)$, and thus the DoPE is defined as
\begin{equation}
    H(t)=e^{-[t-t_{i-1}-\ln(1-\log_2p_i)]}-1.
\end{equation}

We use the average DoPE in an observation time interval (0, $\mathcal{T}$) to study the global effect of pseudonym changes. Referring to \cite{kaul2012real}, the time-average DoPE over (0, $\mathcal{T}$) is given as
\begin{equation}
\label{original_Ht}
   H_{\mathcal{T}}=\frac{1}{\mathcal{T}}\int_0^{\mathcal{T}} H(t)dt.
\end{equation}
To simplify, the area defined by the integral in Eq. (\ref{original_Ht}) can be decomposed into a summation of multiple irregular geometric areas of the same type (i.e., $\widetilde{Q_{1}}$ and $Q_{i}$ for $i\geq2$). More specifically, the decomposition yields
\begin{equation}
\begin{aligned}
H_{\mathcal{T}}& =\frac{\widetilde{Q_{1}}+\Sigma_{i=2}^{\mathcal{N}(\mathcal{T})}Q_{i}}{\mathcal{T}}  \\
&=\frac{\widetilde{Q_{1}}}{\mathcal{T}}+\frac{1}{\mathcal{T}}\sum_{i=2}^{\mathcal{N}(\mathcal{T})}Q_{i},
\end{aligned}
\end{equation}
where $\mathcal{N}(\mathcal{T})=\max\{n| t_n \leq \mathcal{T}\}$ denotes the number of pseudonym changes over a time interval $\mathcal{T}$. The term $\frac{\widetilde{Q_{1}}}T$ will vanish as $\mathcal{T}\to\infty$. Consequently, the time-average DoPE can be rewritten as
\begin{equation}
\label{DoPE}
    \overline{H}=\lim\limits_{{\mathcal{T}}\to\infty}H_{\mathcal{T}} =\frac{\mathcal{N}(\mathcal{T})}{\mathcal{T}}\frac{1}{\mathcal{N}(\mathcal{T})}\sum_{i=2}^{\mathcal{N}(\mathcal{T})}Q_i.
\end{equation}

\begin{figure}[t]
\centering{\includegraphics[width=0.35\textwidth]{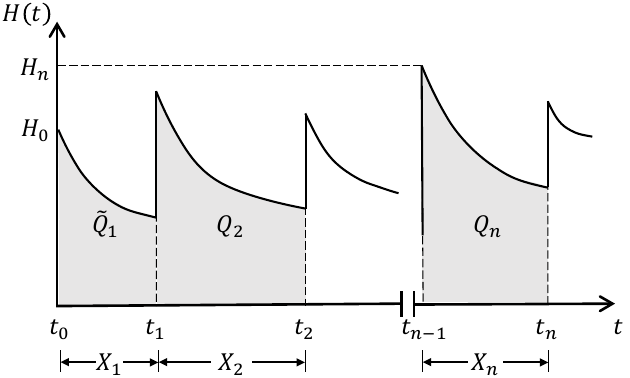}}
\caption{Exponential DoPE for location privacy quantification.}\label{privacy_entropy}
\vspace{-0.5cm}
\end{figure}

Furthermore, for $i\geq2$, the unit area can be obtained by 
\begin{equation}
\label{Qinit}
    \begin{aligned}
        Q_{i}& =\int_{t_{i-1}}^{t_i}\bigl(e^{-[t-t_{i-1}-\ln(1-\log_2p_i)]}-1\bigr)dt  \\
            &=-e^{-[(t_i-t_{i-1})-\ln(1-log_2p_i)]}+1-log_2p_i-(t_i-t_{i-1}).
    \end{aligned}
\end{equation}
For convenience, we define the elapsed time between the $i^{\mathrm{th}}$ and ${(i-1)}^{\mathrm{th}}$ pseudonym changes as
\begin{equation}
\label{elapsed_time}
    X_i=t_i-t_{i-1}.
\end{equation}
Therefore, by substituting Eq. (\ref{elapsed_time}) into Eq. (\ref{Qinit}), the $Q_i$ is converted into
\begin{equation}
\label{unitarea}
    Q_i=(1-e^{- X_i})(1-log_2p_i)-X_i.
\end{equation}
Defining the steady pseudonym change frequency as $\lambda=\lim\limits_{{\mathcal{T}}\to\infty}\frac{\mathcal{N}(\mathcal{T})}{\mathcal{T}}$, the time-average DoPE in Eq. (\ref{DoPE}) can be simplified as
\begin{equation}
\label{simp_DoPE}
    \overline{H}=\lim\limits_{{\mathcal{T}}\to\infty}H_{\mathcal{T}}=\lambda \mathbb{E}[Q_i],
\end{equation}
where $\mathbb{E}[\cdot]$ represents the expectation operator.

The pseudonym changes of VMUs and VTs can be modeled as a Poisson process with rate $\lambda$\cite{luo2023privacy}. Therefore, the elapsed times $X_i$ are independent and identically distributed (iid) exponential random variables with $\mathbb{E}[X]=\frac{1}{\lambda}$\cite{kosta2017age}. According to \cite{kang2017p3}, the successful tracking probability $p_i$ follows a uniform distribution, namely, $p_i{\sim}U(a,b)$. By substituting Eq. (\ref{unitarea}) into Eq. (\ref{simp_DoPE}), the time-average DoPE can be obtained by
\begin{equation}
\label{time_average_H}
\overline{H}=\frac\lambda{\lambda+1}\left(1+\frac1{\ln2}-\frac{b\log_2b-a\log_2a}{b-a}\right)-1.
\end{equation}
Guided by the time-average DoPE, the VMUs and VTs can assess their current privacy levels and decide whether to change pseudonyms. They can request pseudonyms from the local metaverse that they are in and continuously change pseudonyms for identity anonymization. Specifically, to quantify the surplus of privacy preservation through requesting and changing pseudonyms, we formulate $VMU_i$'s utility for consuming all requested pseudonyms over a time slot $t$ in the $j^{\mathrm{th}}$ local metaverse, represented as
\begin{equation}
    u_{j,i}^t=-\varepsilon+(\beta\overline{H}-\delta)R_{j,i}^{t},
\end{equation}
where $\varepsilon$, $\beta$, $\delta$, and $R_{j,i}^{t}$ represent the basic cost of requesting new pseudonyms, the profit of improving privacy protection level by changing each pseudonym under unit DoPE, the additional cost of updating routing table of changing each pseudonym during VT migrations, and the number of pseudonyms that $VMU_i$ actually acquires in the $j^{\mathrm{th}}$ local metaverse, respectively~\cite{luo2023privacy,freudiger2009non}. Therefore, for the $j^{\mathrm{th}}$ local metaverse with a set $\mathcal{I}=\{1,\ldots,i,\ldots,I\}$ of totally $I$ VMUs, the VMU total utility can be obtained by
\begin{equation}
    \mathcal{U}_j^t=\sum_{{i\in}\mathcal{I}}{u_{j,i}^t}=-I\varepsilon+(\beta\overline{H}-\delta)\mathcal{R}_j^t,
\end{equation}
where $\mathcal{R}_j^t=\min\{\mathcal{D}_j^t, G_j^{t}\}$ denotes the total number of pseudonyms that VMUs actually acquire at the beginning of time slot $t$ in the $j^{\mathrm{th}}$ local metaverse. ${D}_j^t$ is the total pseudonym demand of VMUs and $G_j^{t}$ is the number of pseudonyms generated by $LMM_j$ in the $j^{\mathrm{th}}$ local metaverse.

\subsection{Inventory Theory-based Social Welfare Formulation}
The $LMM_j$ periodically generates $G_j^{t}$ pseudonyms based on the observed past pseudonym demands at the beginning of each time slot $t$ and then distributes them to serve VMUs and VTs upon receiving pseudonym requests in the local metaverse. To better promote the privacy-preserving performance of this process within the whole metaverse, we jointly investigate the utilities of VMUs and the LMM in this paper. The Newsvendor model\cite{luo2023privacy} is a crucial component of stochastic inventory theory, which can help managers make informed decisions to maximize their surpluses. Therefore, we employ the Newsvendor model to investigate the optimization problem of pseudonym generation for the LMM. For $LMM_j$, during a time slot $t$, the pseudonym generation incurs generation costs due to computational consumption, while the pseudonym distribution yields provision profits for VMU privacy enhancement\cite{luo2023privacy}. Moreover, if $G_j^t < D_j^t$, the redundant pseudonyms must be retained in the pseudonym pool for a specific duration, incurring storage costs. Conversely, if $D_j^t < G_j^t$, $LMM_j$ faces penalties because of not satisfying the pseudonym requirements of VMUs\cite{luo2023privacy}. Specifically, the utility of $LMM_j$ is represented as
\begin{dmath}
\label{utility}
    U_j^{t}=-g G_j^{t}+(p_0-c) \mathcal{R}_j^t-h\max\left\{\left(G_j^{t}-\mathcal{D}_j^t\right),0\right\}-r\max\left\{\left(\mathcal{D}_j^t-G_j^{t}\right),0\right\},
\end{dmath}
where $g$, $p_0$, $c$, $h$, and $r$ represent the cost of generating each pseudonym, the profit of supplying per pseudonym to VMUs, the communication overhead of distributing each pseudonym, the cost of storing each pseudonym, and the penalty for each unit of unmet pseudonym demand, respectively\cite{luo2023privacy}. Note that $p_0$ is higher than $c$ so that the $LMM_j$ can make profits through providing pseudonyms. 

By considering the utilities of both VMUs and $LMM_j$, we can formulate the social welfare to study the network utilities that reflect the performance of privacy protection in the $j^{\mathrm{th}}$ local metaverse. Covering a whole time period $\mathbf{T}$, the social welfare within the $j^{\mathrm{th}}$ metaverse can be expressed by
\begin{equation}
    {SW}_j^t=(\mathcal{U}_j^t+U_j^{t}), {t\in}\mathbf{T}.
\end{equation}
Here, $\mathbf{T}=\{1,\ldots,t,\ldots,T\}$ means that the whole time period is divided into $T$ equal time slots to recycle the pseudonym management processes\cite{zhang2023towards}. Additionally, the number of pseudonyms generated by each LMM should not exceed the maximum $G_{max}^t$ because of the computation limitations, and the total number within local metaverses cannot exceed the threshold as TA in the cloud layer can only register a specific number of pseudonym from the main chain within time slot $t$. Consequently, the optimization problem of pseudonym management can be transformed into maximizing the overall social welfare, formulated as
\begin{equation}
    \begin{split}
    \textbf{Problem 1:}\:&\sum_{{j\in}\mathcal{J}} \max\:{SW}_j^t  \\
    &\:\:\text{s.t.}\:\: {0 \leq G_j^{t} \leq G_{max}^t},\\
    &\quad\:\:\:\:\: \sum_{{j\in}\mathcal{J}}G_j^t\leq\theta t,\\
    &\quad\:\:\:\:\: 0 < c < p_0,\\
    &\quad\:\:\:\:\: 0 < \delta < \beta\overline{H}.
    \end{split}
\end{equation}
The set $\mathcal{J}=\{1,\ldots,j,\ldots,J\}$ represents that the global metaverse consists of $J$ local metaverses, while $G_{max}^t$ and $\theta$ refer to the maximum number of pseudonyms that can be generated and the upper limit rate of registering pseudonym certificates by the TA within $t$, respectively.

\begin{theorem}
    The independent social welfare in a certain local metaverse ${SW}_j^t$ can reach its maximum.
\end{theorem}

\begin{proof}
    The utility in Eq. (\ref{utility}) can be rewritten as $U_{j}^{t}=-g G_j^{t}+(p_0-c)\min\{\mathcal{D}_j^t,G_j^{t}\}-h(G_j^{t}-\mathcal{D}_j^t)^{+}-r{(\mathcal{D}_j^t-G_j^{t})}^{+}$,
    where $x^+=\max(x,0)$. By exploiting the time-varying characteristic of $U_{j}^{t}$, we further transform the utility into
    \begin{dmath}
    U_j^{t}=(p_0-c+r-g)G_j^{t}-(p_0-c+r+h)\left(G_j^{t}-\mathcal{D}_j^t\right)^{+}-r\mathcal{D}_j^t
    =(p_0-c+r-g)G_j^{t}-(p_0-c+r+h) \int_0^{G_j^{t}}\left(G_j^{t}-\mathcal{D}_j^t\right)f\left(\mathcal{D}_j^t\right)d\mathcal{D}_j^t-r\mathcal{D}_j^t.
    \end{dmath}
    Therefore, the social welfare in the $j^{\mathrm{th}}$ local metaverse in time slot $t$ is given by
    \begin{dmath}
        {SW}_j^t=[-I\varepsilon+(p_0-c+r-g+\beta\overline{H}-\delta)]G_j^{t}-(p_0-c+r+h+\beta\overline{H}-\delta)\int_0^{G_j^{t}}\left(G_j^{t}-\mathcal{D}_j^t\right)f\left(\mathcal{D}_j^t\right)d\mathcal{D}_j^t-r\mathcal{D}_j^t.
    \end{dmath}
    By taking the first-order and second-order derivatives of ${SW}_j^{t}$ with respect to $G_j^{t}$, we can obtain
    \begin{dmath}
        \frac{\partial {SW}_j^t}{\partial G_j^{t}}=(p_o-c+r-g+\beta\overline{H}-\delta) -(p_0-c+r+h+\beta\overline{H}-\delta)F(G_j^{t}),
    \end{dmath}
    \begin{equation}
        \frac{\partial^2{SW}_j^t}{\partial {G_j^{t}}^2}=-(p_0-c+r+h+\beta\overline{H}-\delta)f(G_j^{t})<0,
    \end{equation}
    where $F(\cdot)$ denotes the Cumulative Distribution Function (CDF) of $\mathcal{D}_j^t$. Notably, the first-order derivative has a unique zero point, and the second-order derivative is negative, indicating that the social welfare $SW_j^{t}$ is strictly concave. According to the Newsvendor model, the maximum social welfare in the $j^{\mathrm{th}}$ local metaverse is achieved when $\frac{\partial {SW}_j^t}{\partial G_j^{t}} =0$, namely,
    \begin{equation}
        {G_j^{t}}^*=F^{-1}(\frac{p_o-c+r-g+\beta\overline{H}-\delta}{p_o-c+r+h+\beta\overline{H}-\delta}).
    \end{equation}
    Here, $G_j^{{t}*}$ is the optimal number of pseudonyms generated by $LMM_j$. Therefore, when $LMM_j$ generates ${G_j^{t}}^*$ at time slot $t$, the independent social welfare in the $j^{\mathrm{th}}$ local metaverse can reach its maximum.
\end{proof}

Although the independent social welfare in each local metaverse can reach its maximum, the constraint of pseudonym generation among LMMs complicates the problem of maximizing the overall social welfare. To be specific, the TA cannot register a large number of pseudonyms in a short time, and thus the total number of generated pseudonyms of LMMs in a given time slot should not exceed the predefined maximum. Moreover, due to the variable pseudonym demands of VMUs and the uncertain communication overhead of transmitting pseudonyms within each local metaverse, determining how LMMs collaborate to achieve the optimal pseudonym generation set ${\mathbb{G}^{t}}^* = \{{G_j^{t}}^*\}$ still poses a considerable challenge.

\section{Solution: MADRL-based Pseudonym Generating Strategy}
\label{Edge_MADRL}
To address the aforementioned challenges, we model the pseudonym generation by multiple LMMs as a Partially Observable Markov Decision Process (POMDP)\cite{zhang2023disturb}. According to the properties of POMDP, we adopt an MADRL algorithm based on edge learning technology\cite{xu2023edge} to resolve this problem, of which the details are presented as follows.

\subsection{POMDP for Multi-agent Pseudonym Generation}
\subsubsection{State space} The entire period of pseudonym changes is segmented into equal time steps. At the beginning of time step $t$, the agent ${LMM}_j$ generates $G_j^{t}$ pseudonyms to satisfy the pseudonym requirements within the $j^{\mathrm{th}}$ locals metaverse. After $t$, the VMUs distributed in each local metaverse consume all requested pseudonyms to enhance their respective location privacy, and then ${LMM}_j$ re-generate pseudonyms to meet the following pseudonym demands. Since LMMs can only make decisions according to the previous observations of the environment, we formulate the pseudonym generation as a POMDP. For each $LMM_j$, we define the observation $o_j^t$ at the current decision step $t$ as a union of past $L$-step observations, which is given by
\begin{equation}
    o_j^t\triangleq\{c_j^{t-L},Q_j^{t-L},\mathcal{D}_j^{t-L},\ldots,c_j^{t-1},Q_j^{t-1},\mathcal{D}_j^{t-1}\},
\end{equation}
where $c_j^{t}$, $Q_j^{t}$, and $\mathcal{D}_j^t$ are the average communication overhead between ${LMM}_j$ and VMUs, the periodic pseudonym overproduction, and pseudonym demands of VMUs at time step $t$ in the $j^{\mathrm{th}}$ local metaverse, respectively. $Q_j^{t}=G_j^t-\mathcal{D}_j^t (t\in\{{t-L},\ldots,{t-1}\})$. Consequently, the observation space is defined as the aggregation of observations of all LMMs, denoted as $\boldsymbol{o}^t=\{o_{1}^t,\ldots,o_{j}^t,\ldots,o_{J}^t\}$.

\subsubsection{Action space} In vehicular edge metaverses, we define an action of $LMM_j$ generating pseudonyms at the beginning of $t$ within its local metaverse as $a_j^t=\{G_j^t\}$\cite{zhang2023towards}. Hence, the action space is a set containing the actions of each agent, represented by $\boldsymbol{a}^t=\{a_{1}^t,\ldots,a_{j}^t,\ldots,a_{J}^t\}$.

\subsubsection{Reward} According to the current observation state $o_j^t$, the edge learning-based $LMM_j$ selects an action $a_j^t$ to gain the reward, and then $o_j^t$ transitions to $o_j^{t+1}$\cite{xu2023edge}. The reward for pseudonym generation of each agent at time slot $t$ can be defined as $R(o_j^t,a_j^t)={SW}_j^t$.
In the vehicular edge metaverses, maximizing social welfare is the common goal of LMMs. Therefore, the reward function is the sum of LMM's reward at time slot $t$, defined as
\begin{equation}
    R(\boldsymbol{o}^t,a^t) = \begin{cases}
        \sum_{{j\in}\mathcal{J}}R(o_j^t,a_j^t), & \sum_{{j\in}\mathcal{J}}G_j^t \leq \theta t, \\\\
        0, & \mathrm{otherwise}. \hfill
    \end{cases}
\end{equation}
For convenience, we abbreviate $R(o_j^t,a_j^t)$ and $R(\boldsymbol{o}^t,a^t)$ as $R_j^t$ and $R^t$, respectively.

\subsection{Algorithm Details}
\begin{algorithm}[t]  
\small
\caption{MAPPO Algorithm for Pseudonym Generation in Vehicular Edge Metaverses}\label{algorithm}
Initialize maximum episodes $E$, maximum time steps $T$ in an episode, maximum epochs $K$, and batch size $B$ \;
Initialize actor $\pi_{{\theta}_j}$, $\pi_{{\theta}_j}^{old}$ and critic $Q_{{\omega}_j}$, $Q_{\overline{\omega}_j}$ \;
\For{Episode $e=1,2,\ldots,E $}
{       
    Reset Pseudonym Generation Environment $PGEnv$ and replay buffer $\mathcal{B F}$\;
    \For{Time step $t=1,2,\ldots,T$}
    {   
        Each agent $LMM_j$ observes $o_j^t$ and selects an action $a_j^t$ according to its current actor policy $\pi_{{\theta}_j}^{old}$\;
        Get the reward $r_t$ and update $\boldsymbol{o}^t$ into $\boldsymbol{o}^{t+1}$\;
    }
    Each agent $LMM_j$ obtains a trajectory $\tau_j=\{o_j^t,a_j^t,R_j^t,o_j^{t+1}\}_{t=1}^{T}$\;
    Compute $\{\hat Q_j(\boldsymbol{o}^t,\boldsymbol{a}^t)\}_{t=1}^{T}$ according to Eq. (\ref{qhat})\;
    Compute advantages $\{A_j(\boldsymbol{o}^t,\boldsymbol{a}^t)\}_{t=1}^{T}$ according to Eq. (\ref{advantage_function})\;
    Store data $\{\{o_j^t,a_j^t,\hat Q_j[\boldsymbol{o}^t,\boldsymbol{a}^t],A_j[\boldsymbol{o}^t,\boldsymbol{a}^t]\}_{j=1}^{J} \}_{t=1}^{T}$ into replay buffer $\mathcal{B F}$\;
    \For{Epoch $k=1,2,\ldots,K$}
    {
        Shuffle the data order in $\mathcal{B F}$\;
        \For {$l=1,2,\ldots,\frac{T}{B}-1$}
            {
                Sample a mini-batch of data $d_l$ with a size $B$ from $\mathcal{B F}$, where $d_l=\{[o_j^m,a_j^m,\hat Q_j[\boldsymbol{o}^m,\boldsymbol{a}^m],A_j[\boldsymbol{o}^m,\boldsymbol{a}^m]_{j=1}^J\}_{m=1+Bl}^{B(l+1)}$ 
                \For {$j=1,2,\ldots,J$}
                {
                    $\Delta\theta_j=\frac{1}{B}\sum_{m=1}^{B}\{\nabla_{\theta_j}G[r_m(\theta_j),$\\
                    $A_{j}(\boldsymbol{o}_m,\boldsymbol{a}_m)]\}$ \\
                    $\Delta\omega_{j}=\frac1B\sum_{m=1}^{B}\{\nabla_{\omega_{j}}(\hat{Q}_{j}(\boldsymbol{o}_m,\boldsymbol{a}_{m})$ \\
                    $-Q_{\omega_j}(\boldsymbol{o}_m,\boldsymbol{a}_m))^2\}$ \\
                    Apply gradient ascent to update actor parameter $\theta_j$ using $\Delta\theta_j$\;
                    Apply gradient descent to update critic parameter $\omega_j$ using $\Delta\omega_j$\;
                }
            }
    }
    Update $\theta_j^{(old)}\leftarrow\theta_j$ and $\overline{\omega}_j\leftarrow\omega_j$ for each $LMM_j$\;
}
\end{algorithm}
We adopt the actor-critic framework and employ the Multi-agent Proximal Policy Optimization (MAPPO) approach with centralized training and decentralized execution for policy iteration\cite{junlong2023MADRL,du2023maddpg}. The hyperparameters of $LMM_j$'s policy ${\pi}_{\theta_j}$ and collective policy ${\pi}_{\boldsymbol{\theta}}$ are denoted as ${\theta_j}$ and ${\boldsymbol{\theta}}$, respectively. Here, $\boldsymbol{\theta}=\{\theta^1,\ldots,\theta^N\}$. Therefore, the objective of MAPPO can be formulated as
\begin{equation}
\label{20}
    \max_{\boldsymbol{\theta}}{\mathbb{E}}_{\pi_{\boldsymbol{\theta}}^{old}}\bigg\{\sum_{{j\in}\mathcal{J}}G[r(\theta_j),A_{\pi_{\boldsymbol{\theta}}^{old}}(\boldsymbol{o},\boldsymbol{a})]\bigg\},
\end{equation}
where the current policy of LMMs ${\pi_{\boldsymbol{\theta}}^{old}}$ is a differentiable function with hyperparameter $\boldsymbol{\theta}^{old}$ and $A_{\pi_{\boldsymbol{\theta}}^{old}}(\boldsymbol{o}^t,\boldsymbol{a}^t)$ is the advantage function\cite{junlong2023MADRL}. To achieve feasible implementation of algorithm training, the objective function can be calculated by the expectation over a batch of samples. Specifically, the policy network of $LMM_j$ is updated through gradient ascent~\cite{zhang2023learning}, expressed as
\begin{equation}
\label{gradient}
    \Delta\theta_j=\nabla_{\theta_j}\hat{\mathbb{E}}_t \bigg\{G[r_t(\theta_j),A_j(\boldsymbol{o}_t,\boldsymbol{a}_t)] \bigg\}.
\end{equation}
Here, ${\mathbb{E}}_t \{ \cdot \}$ is the sample average and $r_{t}(\theta_j)=  \frac {\pi _ {\theta_j }(\boldsymbol{a}_ {t}|\boldsymbol{o}_ {t})}{\pi _ {\theta_j^{old} }(\boldsymbol{a}_ {t}|\boldsymbol{o}_ {t})}$ is the importance ratio\cite{zhang2023learning}. $A_j(\boldsymbol{o}^t,\boldsymbol{a}^t)$ is the estimation of $A_{\pi_{\boldsymbol{\theta}}^{old}}(\boldsymbol{o}^t,\boldsymbol{a}^t)$, which can be calculated based on the Generalized Advantage Estimation (GAE) method\cite{zhang2023learning}. We also leverage a clip mechanism\cite{junlong2023MADRL} to constrain the policy updates in Eq. (\ref{gradient}), which can be further expressed by
\begin{equation}
\begin{split}
    G[r_t(\theta_j),A_j(\boldsymbol{o}_t,\boldsymbol{a}_t)]=\min \bigg\{&(r_{t}(\theta_{j}) A_j(\boldsymbol{o}^t,\boldsymbol{a}^t),\\ &g_{clip}[\epsilon, A_j(\boldsymbol{o}^t,\boldsymbol{a}^t)] \bigg\},
\end{split}
\end{equation}
where 
\begin{equation}
    g_{clip}(\epsilon, A)=  \begin{cases}1-\epsilon,\:A<0,\\
1+\epsilon,\:A\geq 0.\end{cases}
\end{equation}
Note that $\epsilon\in[0,1]$ is the clipping parameter\cite{junlong2023MADRL}.

In this paper, we calculate the advantage estimation in the form of state-action value function\cite{junlong2023MADRL}. To tackle the problem of not knowing the impact that $LMM_j$ generates a certain number of pseudonyms on the total reward (i.e., the multi-agent credit assign problem), we also use a counterfactual baseline \cite{foerster2018counterfactual} to calculate the estimates, given by
\begin{equation}
\label{advantage_function}
    A_j(\boldsymbol{o}^t,\boldsymbol{a}^t)=\hat{Q}_j(\boldsymbol{o}^t,\boldsymbol{a}^t
    )-b(\boldsymbol{o}^t,\boldsymbol{a}_{-j}^t).
\end{equation}
Here, $b(\boldsymbol{o}^t,\boldsymbol{a}_{-j}^{t})=\sum_{a_{j}^{t}}{\pi_{\theta_j}^{old}}(a_{j}^{t}|o_{j}^{t})Q_{\omega_{j}}[\boldsymbol{o}^t,(\boldsymbol{a}_{-j}^{t},a_{j}^{t})]$ is the counterfactual baseline, and $\boldsymbol{a}_{-j}^{t}$ is the joint action of other agents except $LMM_j$\cite{junlong2023MADRL}. $\hat Q_j(\boldsymbol{o}^t,\boldsymbol{a}^t)$ is the estimation of state-action value function, calculated by
\begin{equation}
\label{qhat}
    \hat Q_j(\boldsymbol{o}^t,\boldsymbol{a}^t)=Q_{\bar\omega_j}(\boldsymbol{o}^t,\boldsymbol{a}^t)\\+\delta_t+(\gamma\lambda_{gae})\delta_{t+1}+\cdots+(\gamma\lambda_{gae})^T\delta_T,
\end{equation}
where the TD error $\delta_t=R_t+\gamma Q_{\bar{\omega}_j}(\boldsymbol{o}^{t+1},\boldsymbol{a}^{t+1})-Q_{\bar{\omega}_j}(\boldsymbol{o}^{t},\boldsymbol{a}^{t})$, $Q_{\bar\omega_j}(\boldsymbol{o}^t,\boldsymbol{a}^t)$ is the centralized critic of $LMM_j$, $\gamma$ is the discount factor, and $\lambda_{gae}$ is the decay factor\cite{junlong2023MADRL}.
\begin{figure*}[t]
\centering{\includegraphics[width=0.8\textwidth]{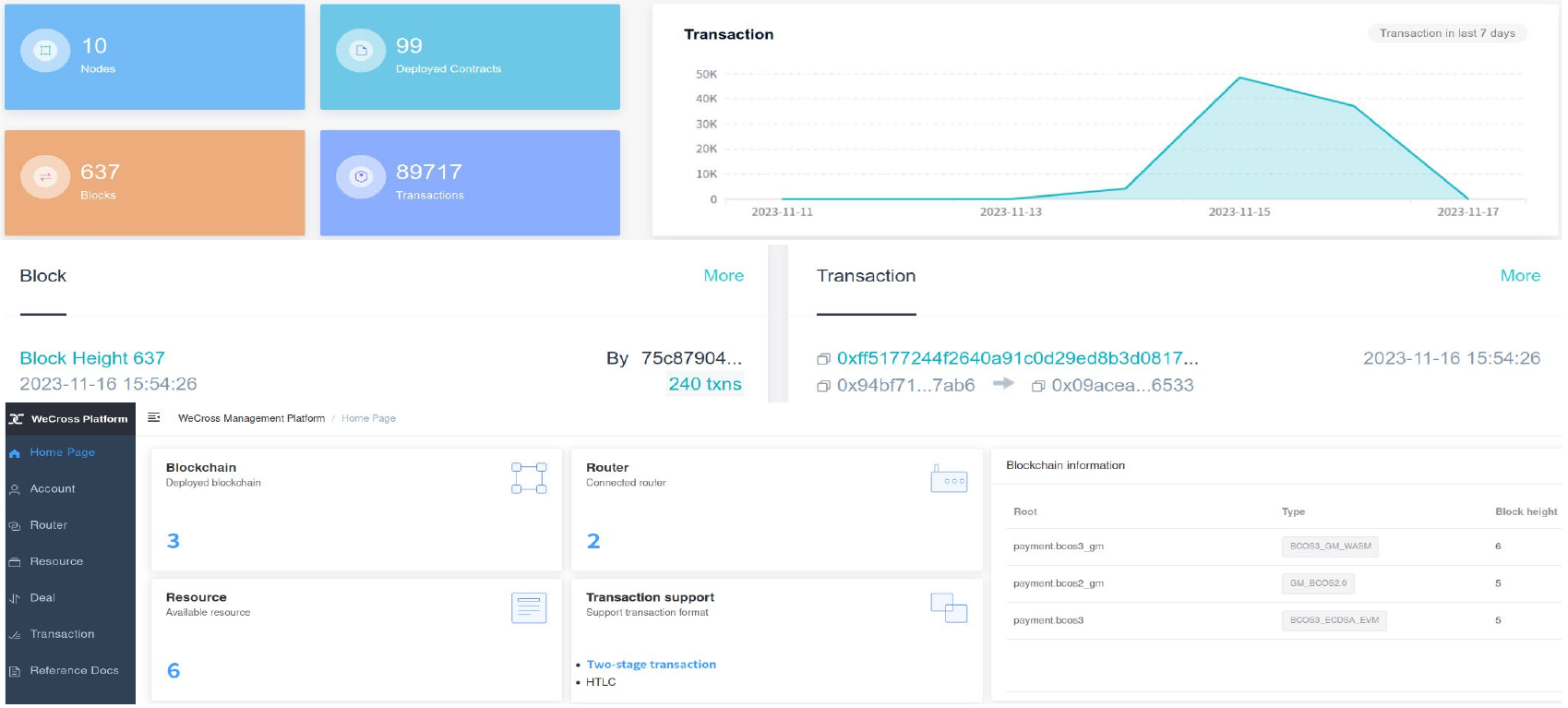}}
\caption{The blockchain system building on FISCO BCOS and WeCross platforms. The figure shows the status of the single-chain and cross-chain system with block and transaction information.}\label{cross-chain_platform}
\vspace{-0.5cm}
\end{figure*}

The complete pseudo-code of the MAPPO algorithm for pseudonym generation is presented in \textbf{Algorithm \ref{algorithm}}. During the collecting process, each agent $LMM_j$ continually interacts with the pseudonym generation environment to select an action with the current policy $\pi_{\theta_j^{(old)}}$ to obtain the trajectories. Then, the Q-function and advantage estimation are calculated by the target critic $Q_{\bar\omega_j}$. After that, the data $\left\{\left\{o_{j}^{t},a_{j}^{t},\hat{Q}_{j}(\boldsymbol{o}^{t},\boldsymbol{a}^{t}),A_{j}(\boldsymbol{o}^{t},\boldsymbol{a}^{t})\right\}_{\mathrm{j}=1}^{J}\right\}_{t=1}^{T}$ will be stored in the replay buffer. During the training process, the optimizer randomly samples experiences from the replay buffer to update the network parameters\cite{zhang2023learning} in each epoch. Then, the network parameters $\theta_j^{(old)}$ and $\overline{\omega}_j$ are updated to $\theta_j$ and $\omega_j$, respectively. For each agent $LMM_j$, the critic parameter $\omega_j$ is updated by minimizing the loss function $L(\omega_j)=(\hat{Q}_j(\boldsymbol{o}^t,\boldsymbol{a}^t)-Q_{\omega_j}(\boldsymbol{o}^t,\boldsymbol{a}^t))^2$. The time complexity of the employed MAPPO algorithm hinges on the multiplication operations within multiple fully-connected deep neural networks \cite{zhang2023learning,zhang2023disturb}. This complexity is denoted by $\mathcal{O}\left(\sum_{f=1}^{F+1} \xi_{f} \xi_{f-1}\right)$, where $\xi_{f}$ signifies the number of neural units in the $f^{th}$ layer and $F$ represents the total number of hidden layers.

\section{Performance Evaluation}
\label{Performance Evaluation}
\subsection{Security Analysis}
The cross-metaverse empowered dual pseudonym management framework has a positive effect on privacy protection, satisfying the following anticipated security requirements.

\begin{enumerate}[1)]
\item \emph{Reliable privacy protection:} The proposed framework ensures \textbf{\textit{anonymity}} in vehicular edge metaverses by allowing both VMUs and VTs quickly acquire pseudonyms in the local metaverse to conceal their true identities\cite{luo2023privacy,xu2021efficient}. The VMUs and VTs can utilize the DoPE to evaluate their current privacy levels and decide whether to change pseudonyms. Even if attackers have eavesdropped a safety message including a VMU pseudonym identity, they cannot link it to the associated VT. This is because when VMU and VT change pseudonyms synchronously, attackers will confuse the identity of the target VMU with that of other VMUs in the same local metaverse~\cite{luo2023privacy}, meaning that the security requirement \textbf{\textit{unlinkability}} is well satisfied.

\item \emph{Data integrity and immutability:}
To improve management efficiency and pseudonym security, we utilize a consortium blockchain and Practical Byzantine Fault Tolerance (PBFT) consensus algorithm in the cross-chain system~\cite{li2020scalable,liu2023reputation}. Given the nature of integrity of hash-based blocks, the subchain in the local metaverse cannot be easily cracked by attackers, thereby guaranteeing pseudonym \textbf{\textit{immutability}}. Furthermore, combined with the notary mechanism, only the authenticated edge servers are permitted to make cross-chain transactions. Hence, only the fully trusted TA, notaries (i.e., LMMs) and verified nodes can access data on the relay chain and other subchains, thus ensuring the \textbf{\textit{conditional traceability}} in vehicular edge metaverses. Moreover, even though a certain subchain in a local metaverse crashes due to a disastrous attack, the relay chain and other subchains can continue to operate because the on-chain data are partially isolated among subchains. This demonstrates the \textbf{\textit{robustness}} of our proposed framework.

\item \emph{Efficient pseudonym management:} 
With the aid of the hierarchical cross-metaverse architecture, the pseudonym management is accelerated by LMMs in the edge layer, significantly reducing the communication delay compared to traditional centralized schemes\cite{kang2017p3}. Moreover, the MADRL algorithm based on edge learning technology enables multiple LMMs to generate pseudonyms swiftly~\cite{xu2023edge}. Therefore, we achieve high \textbf{\textit{Efficiency}} in pseudonym management.
\end{enumerate}

\subsection{Performance Analysis of the Cross-chain Assisted Pseudonym Management Scheme}
\begin{figure}[t]
\vspace{-0.5cm}
\centering{\includegraphics[width=0.33\textwidth]{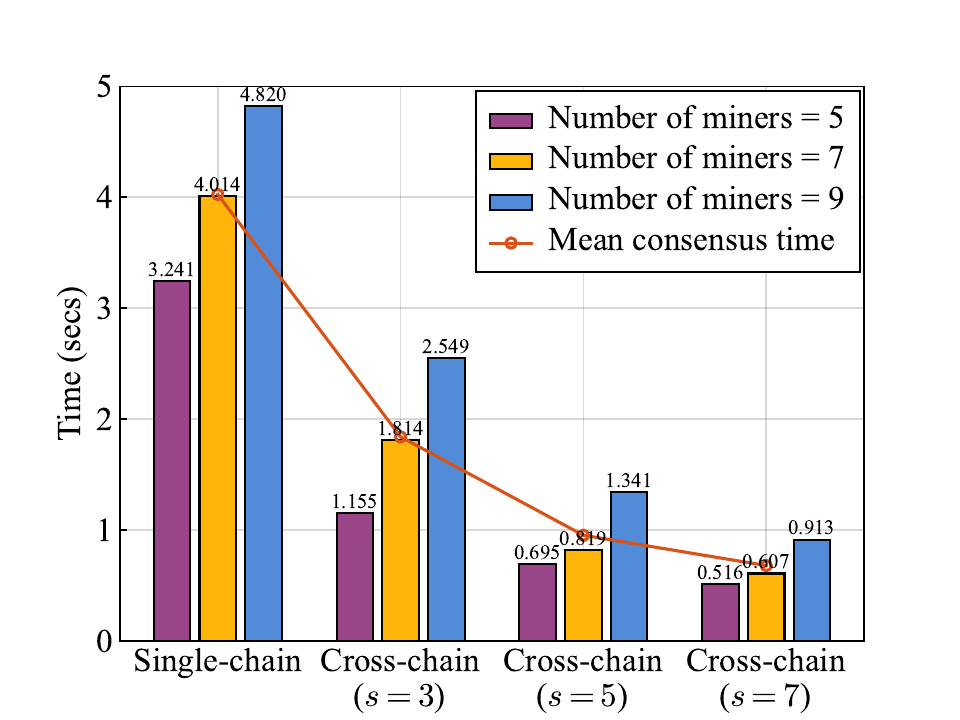}}
\caption{Consensus time comparison between single-chain and cross-chain system.}\label{latency_subchain}
\vspace{-0.5cm}
\end{figure}
To evaluate the proposed cross-chain assisted pseudonym management, we conduct simulation experiments w.r.t. blockchains on the FISCO BCOS platform integrated with a cross-chain platform called WeCross\cite{kangjinbo2023healthcare}, as shown in Fig. \ref{cross-chain_platform}. The simulations are run on an Ubuntu 22.04 system with an Intel Core i7-12700 CPU @2.10 GHz, including 8 GB RAM. By default, the number of pseudonym-related transactions is set to $1000$, and the data size per transaction is set to $1$ KB~\cite{khan2022privacy}. Meanwhile, We employ five subchains in our cross-chain system as the default setting.

Fig. \ref{latency_subchain} shows the consensus time for adding new blocks (e.g., pseudonym registration information) in the single-chain system and the cross-chain system. Focusing on the red solid line, we find that the single-chain system takes the longest mean consensus time to complete $1000$ transactions. In contrast, our cross-chain system with the PBFT consensus significantly reduces the consensus time\cite{feng2023wireless}. As the number of subchains $s$ increases from $3$ to $7$, the average consensus time in our cross-chain system decreases from $1.839$ seconds to $0.679$ seconds, indicating that the consensus efficiency of our cross-chain system is nearly $6\times$ higher than that of the single-chain system when the number of subchains exceeds $7$. Additionally, as the number of miners increases, the time for adding blocks increases because there are more nodes participating in the consensus.

Fig. \ref{latency_trans_num} shows the effects of the number of transactions on consensus time. It is evident that as the number of transactions increases, both the single-chain system and cross-chain system undergo incremental block delays. Nonetheless, the consensus time in our cross-chain system increases smoothly, whereas the single-chain system experiences a much sharper increase. Compared with the single-chain scheme, our proposed schemes exhibit a notable reduction in consensus time by $87.973\%$ when processing $2500$ transactions, showcasing their capability to handle high-throughput scenarios of pseudonym management in vehicular edge metaverses.

Table \ref{request_delay} shows a comparison of pseudonym requesting delay, which encompasses cryptographic operation time, communication delay, and blockchain verification time. According to \cite{kang2017p3}, the cryptographic operation time includes the time of asymmetric encryption and decryption, the time of signature generation and verification, and the time of certificate verification, which are set to $1.86$ ms, $0.94$ ms, $0.93$ ms, $1.11$ ms, $5.42$ ms, respectively. The communication delay comprises the delays between the VMU and edge server, edge server and LMM, and edge server and TA, set to $20$ ms, $5$ ms, and $10$ ms, respectively\cite{kang2017p3}. We can see that for local pseudonym distribution (denoted as \textit{local}) in our cross-chain system, all three types of time are lower compared to the single-chain system. This is because the distributed LMMs close to VMUs can reduce operational complexity and communication latency, and there are fewer workers participating in consensus on the subchain, which can improve verification efficiency. Despite the slightly higher verification time for cross-district pseudonym distribution (denoted as \textit{cross-district}), our proposed scheme is still practical. Under acceptable time overhead, the pseudonyms are only recorded on original subchains when VMUs and VTs migrate among local metaverses in our cross-chain scheme, thus partially isolating sensitive data to enhance secure pseudonym management.

\begin{figure}[t]
\vspace{-0.5cm}
\centering{\includegraphics[width=0.33\textwidth]{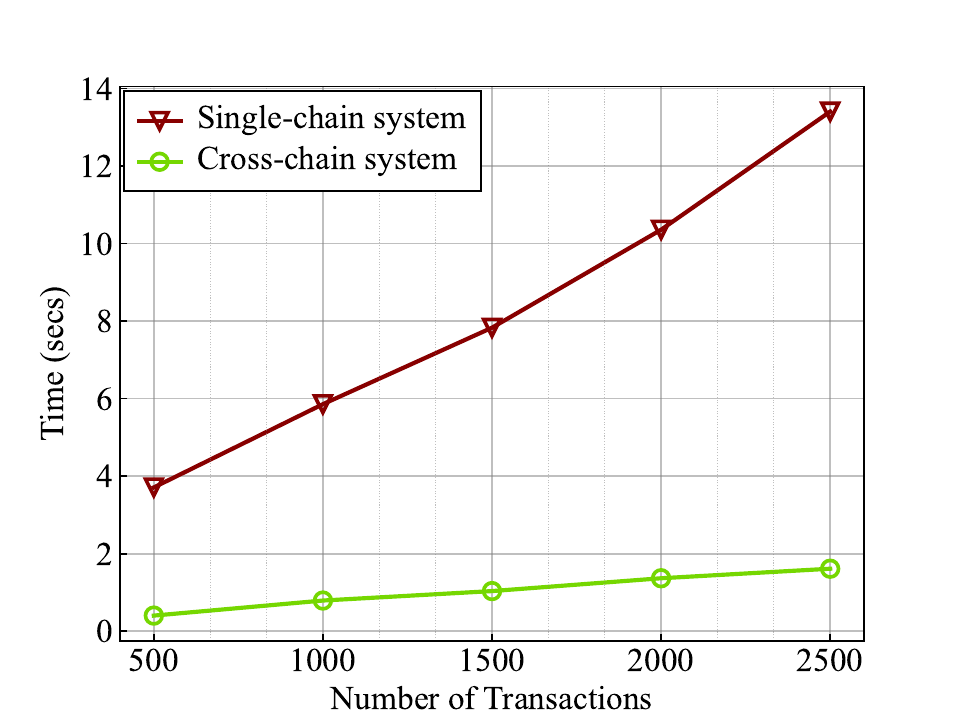}}
\caption{Consensus time corresponding to different numbers of transactions under different system types.}\label{latency_trans_num}

\end{figure}

\begin{table}[t]
\renewcommand{\arraystretch}{1.25} 
\captionsetup{font=footnotesize} 
\caption{Pseudonym requesting delay (milliseconds) comparison.}
\centering
\resizebox{\columnwidth}{!}{%
\begin{tabular}{|c|c|c|cc|}
\hline
\multirow{2}{*}{\textbf{\begin{tabular}[c]{@{}c@{}}System\\ type\end{tabular}}} & \multirow{2}{*}{\textbf{\begin{tabular}[c]{@{}c@{}}Cryptographic\\ operation time\end{tabular}}} & \multirow{2}{*}{\textbf{\begin{tabular}[c]{@{}c@{}}Communi-\\ cation delay\end{tabular}}} & \multicolumn{2}{c|}{\textbf{\begin{tabular}[c]{@{}c@{}}Blockchain \\ verification time\end{tabular}}} \\ \cline{4-5} 
 &  &  & \multicolumn{1}{c|}{\textit{local}} & \textit{cross-district} \\ \hline
\multicolumn{1}{|l|}{Single-chain} & $19$ & $60$ & \multicolumn{2}{c|}{$28$} \\ \hline
Cross-chain & $7$ & $50$ & \multicolumn{1}{c|}{$21$} & $806$ \\ \hline
\end{tabular}\label{request_delay}%
}
\vspace{-0.5cm}
\end{table}
\subsubsection{Parameter setting}
To evaluate the performance of the MAPPO algorithm for pseudonym generation, we investigate a scenario where multiple LMMs generate pseudonyms to meet the pseudonym demands of VMUs in simulation experiments. Specifically, We consider that there are three LMMs and the number of VMUs in each local metaverse $\mathcal{I}$ is set to $[80, 70, 60]$. We assume that the varying pseudonym demands in each local metaverses follow a Poisson distribution with mean $[80, 90, 100]$ per minute, since generally the fewer VMUs in a district, the more pseudonyms are needed for privacy preservation. For simplicity, we consider that the communication overhead follows a uniform distribution, namely, $c_j^{t}\sim{U[0,0.2]}$. Meanwhile, we set default values $\varepsilon=0.1$, $\beta=0.2$, $\delta=0.5$, and $p_0=1.5$. The length of a time slot $t$ is set to $60$ seconds, the maximum number of generated pseudonyms of each LMM $G_{max}^t$ is set to $120$, and the upper limit rate for registering pseudonyms in the global metaverse $\theta$ is set to $5$ per second. The minimum and maximum number of vehicles in a local metaverse are set to $10$ and $160$, respectively. Regarding the configuration of learning-based algorithms, we set $T=120$, $L=3$, $K=15$, and $B=16$. The learning rate of actor and critic is set to $0.001$. The clip parameter $g_{clip}$ is set to $0.2$. The discount factor $\gamma$ is set to $0.99$, and the decay factor $\lambda_{gae}$ is set to $0.95$. Finally, $\xi_{f}$ and $F$ are set to $64$ and $1$, respectively.

\begin{figure}[t]
\vspace{-0.5cm}
\centering{\includegraphics[width=0.33\textwidth]{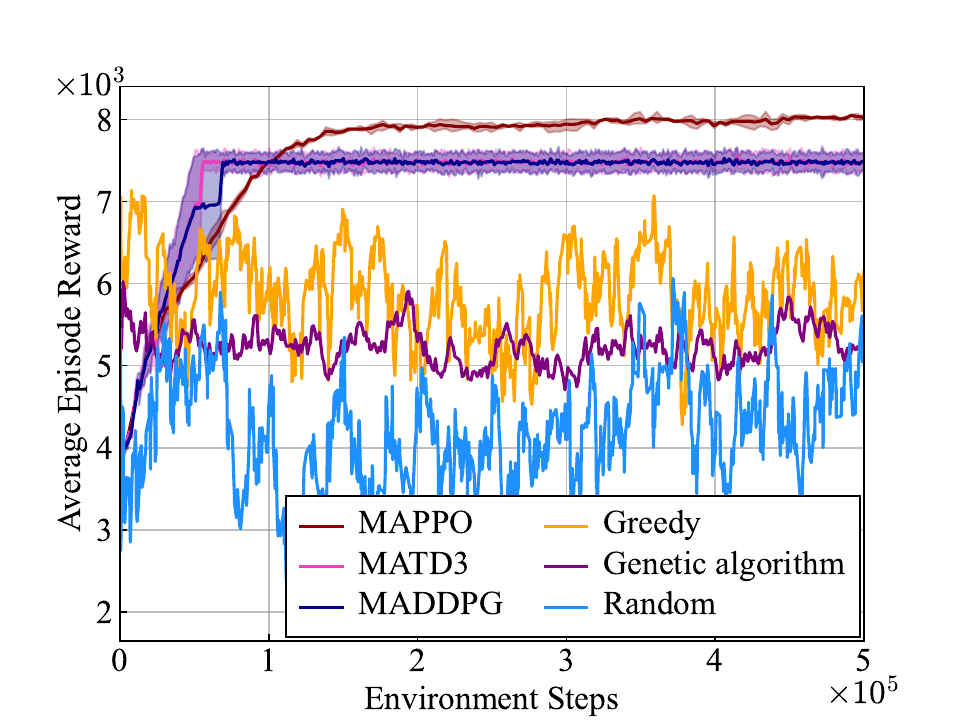}}
\caption{Comparison of average episode reward curves of MAPPO and benchmarks for the pseudonym generation task.}\label{episode_reward}

\end{figure}

\begin{figure}[t]
\centering{\includegraphics[width=0.33\textwidth]{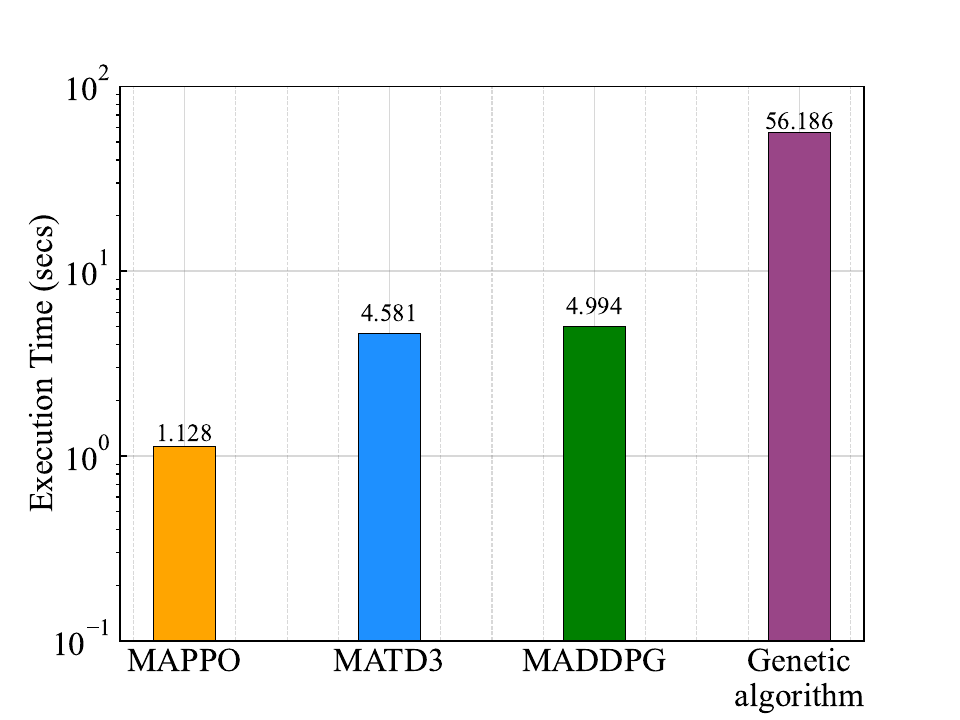}}
\caption{Execution time of obtaining optimal strategy for pseudonym generation under different schemes.}\label{run_time}
\vspace{-0.5cm}
\end{figure}
\subsection{Performance Analysis of the Proposed Pseudonym Generation Method}
\subsubsection{Convergence analysis}
    As shown in Fig. \ref{episode_reward}, we compare the convergence performance of our proposed MAPPO-based scheme with several benchmark approaches, including $\rmnum{1})$ \textit{Multi-agent Deep Deterministic Policy Gradient (MADDPG)}, $\rmnum{2})$ \textit{Multi-agent Twin Delayed Deep Deterministic Policy Gradient (MATD3)}, $\rmnum{3})$ \textit{genetic algorithm}\cite{chaudhary2019pseudonym}, $\rmnum{4})$ \textit{random}, and $\rmnum{5})$ \textit{greedy}. The MADDPG and MATD3 are learning-based algorithms\cite{du2023maddpg}, while the genetic algorithm is a classical heuristic algorithm. In the random scheme, the LMM randomly determines the number of pseudonyms to generate, while in the greedy scheme, the LMM determines the number based on the maximum utility achieved in previous time steps. We can see that the proposed scheme can converge on the maximum reward, outperforming MATD3, MADDPG, genetic algorithm, greedy, and random by $8.7\%$, $8.8\%$, $37.1\%$, $53.3\%$, and $92.9\%$, respectively. Since the pseudonym demands are time-varying, traditional heuristic algorithms fail to reach the convergence value. In Fig. \ref{run_time}, we compare the execution time of each algorithm over $1000$ episodes. The proposed MAPPO-based scheme is found to be four times faster than other learning-based algorithms and nearly $50\times$ faster than the heuristic. Consequently, our MAPPO-based solution requires less training time and performs better, highlighting its competence in pseudonym generation for vehicular edge metaverses.
\begin{figure}[t]
\vspace{-0.5cm}
\centering{\includegraphics[width=0.33\textwidth]{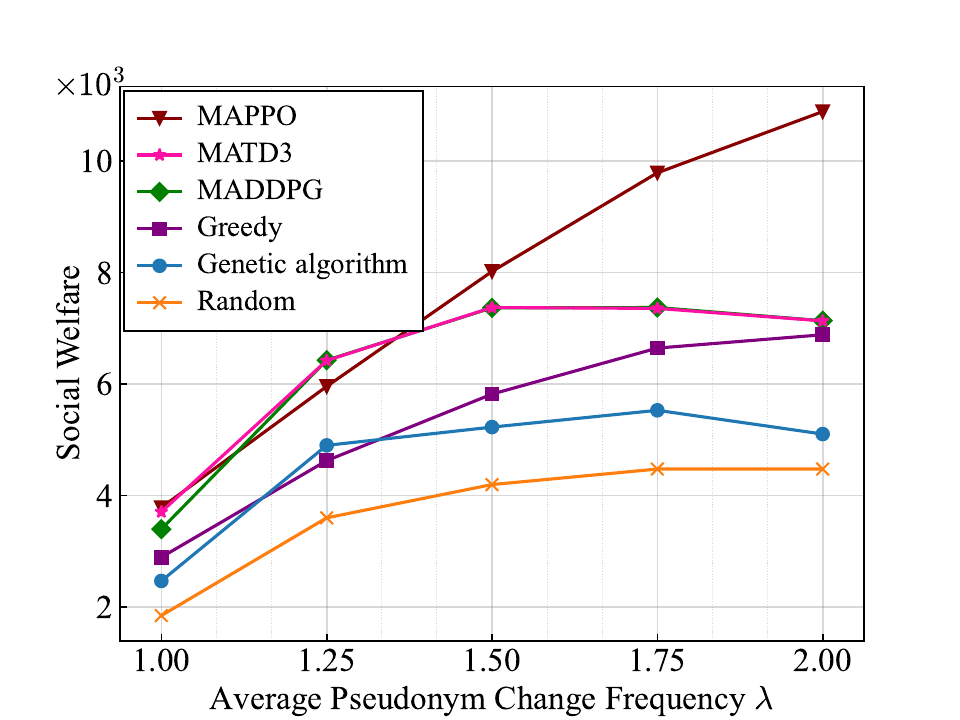}}
\caption{Social welfare under different average pseudonym change frequencies.}\label{SW_lambda}

\end{figure}

\begin{figure}[t]
\centering{\includegraphics[width=0.33\textwidth]{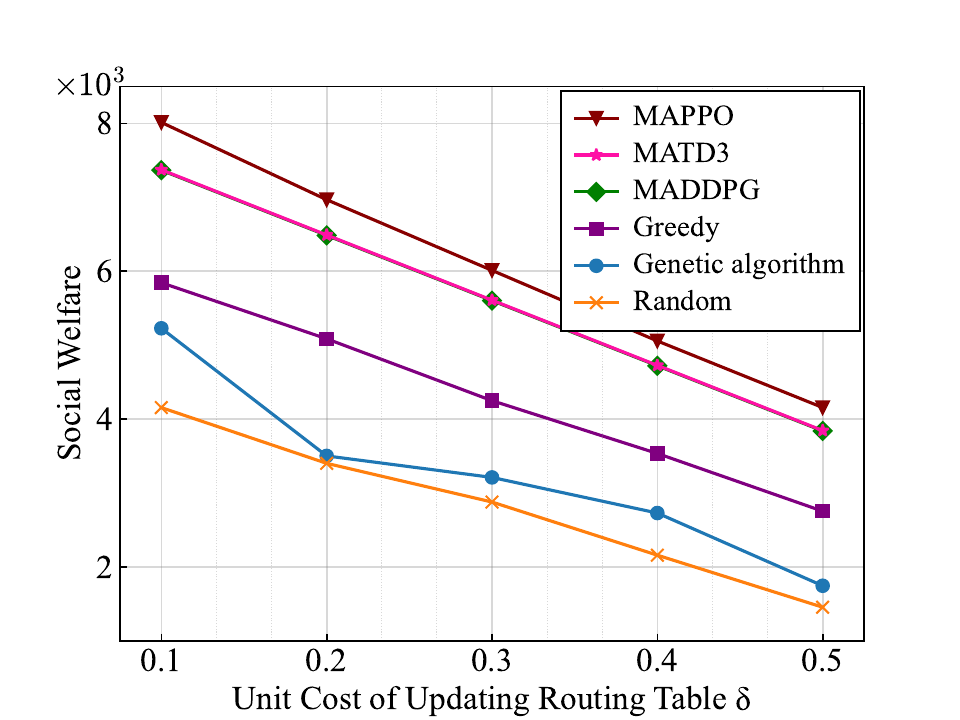}}
    \caption{Social welfare under different unit costs of updating routing table.}\label{SW_delta}
\vspace{-0.5cm}
\end{figure}
\subsubsection{Utility analysis}
Fig. \ref{SW_lambda} shows the impact of pseudonym change frequency on social welfare. From Eq. (\ref{time_average_H}) we know that the time-average DoPE is positively correlated with average pseudonym change frequency $\lambda$. It is observed that our MAPPO-based scheme achieves the highest social welfare under different $\lambda$, except for $1.25$. Furthermore, with the $\lambda$ ranging from $1.75$ to $2$, the social welfare decreases under the methods of MATD3, MADDPG, and the genetic algorithm. This phenomenon implies that these baseline schemes are not well-suited for the complex vehicular edge metaverse with ever-changing pseudonym demands. Therefore, our proposed method exhibits greater feasibility when applied in metaverses.

Fig. \ref{SW_delta} shows the effect of different unit costs of updating routing table on the social welfare under various schemes. On account of the dynamic migrations in vehicular edge metaverses, changing both VMU and VT pseudonyms can incur a heterogeneous cost of updating routing table within different local metaverses\cite{freudiger2009non}. We find that our proposed scheme consistently achieves maximum social welfare under each unit cost of updating routing table $\delta$. Therefore, the proposed schemes exhibit greater adaptability to the fluctuating network condition within vehicular edge metaverses.

\section{Conclusion}
\label{Conclusion}
In this paper, we examined the transformative potential of privacy-preserving pseudonym management in vehicular edge metaverses. Considering the dynamic nature of VMU and VT migrations, we presented a cross-metaverse empowered dual pseudonym management framework, in which the global metaverse consists of multiple local metaverse collaborating for efficient VMU and VT pseudonym management. Then, we integrated the cross-chain technology into our framework, with its decentralized architecture facilitating secure pseudonym distribution and revocation. Furthermore, we proposed an analytical metric named DoPE to assess the degree of privacy protection after pseudonym changes for VMUs and VTs. Combining DoPE with inventory theory, we formulated the optimization problem of pseudonym generation in vehicular edge metaverses. Additionally, due to the ever-changing pseudonym demands within multiple local metaverses, we employed an MADRL algorithm based on edge learning technology to achieve high-efficiency and cost-effective pseudonym generation. Finally, numerical results demonstrated the effectiveness and feasibility of our proposed framework in vehicular edge metaverses. For future work, we will further explore the applications of advanced optimization tools and the proposed metric across various domains in vehicular metaverses, such as pseudonym changes and exchanges, among others.

\bibliographystyle{IEEEtran}
\bibliography{reference}

\end{document}